%% file: SpeedThroughWordRepresentations.tex
\documentclass[runningheads,dvipsnames,anonymous]{llncs}

\newcommand{\longversion}[1]{#1}
\newcommand{\shortversion}[1]{}
\usepackage{pdflscape}

\usepackage{lineno}
\nolinenumbers

\input{shortcuts}

\newcommand{\myref}[1]{\autoref{#1}}
\newif\ifinappendix

\newcommand{\applabel}[1]{
\shortversion{\ifinappendix\hypertarget{app:#1}{}\else\hyperlink{app:#1}{$(*)$}\fi}
\label{#1}}

\newcommand{\lborder}{\textsf{L}}
\newcommand{\rborder}{\textsf{R}}
\newcommand{\intervals}{\mathcal{I}}
\newcommand{\position}{\text{\textsf{ind}}}

\usepackage{comment}
\usepackage[T1]{fontenc}
\usepackage{graphicx}
\usepackage{xcolor}
\usepackage{afterpage}
\definecolor{LB}{RGB}{188,218,234}
\definecolor{DB}{RGB}{33,121,180}
\definecolor{LR}{RGB}{252,179,179}
\definecolor{DR}{RGB}{233,76,78}
\definecolor{DG}{RGB}{51,159,43}
\definecolor{DO}{RGB}{255,126,0}
\usepackage{enumitem}
\usetikzlibrary{arrows.meta,patterns,ipe,backgrounds}
\newtheoremrep{thm}[theorem]{Theorem}
\newtheoremrep{lem}[lemma]{Lemma}
\newtheoremrep{obs}[observation]{Observation}
\newtheoremrep{cor}[corollary]{Corollary}

\begin{document}
\title{Determining Factorial Speed Fast}
%
%\titlerunning{Abbreviated paper title}
% If the paper title is too long for the running head, you can set
% an abbreviated paper title here
%
\shortversion{\author{A. Anonymous}}
\longversion{\author{Zhidan Feng\inst{1}\orcidID{0000-0002-3364-5396} \and
Henning Fernau\inst{2}\orcidID{0000-0002-4444-3220} \and
Pamela Fleischmann\inst{3}\orcidID{0000-0002-1531-7970} \and 
Philipp Kindermann\inst{2}\orcidID{0000-0001-5764-7719} \and
Silas Cato Sacher\inst{2}\orcidID{0009-0004-6850-1298}
}}
\longversion{\authorrunning{Z. Feng, H. Fernau, P. Fleischmann, P. Kindermann, %K. Mann, 
S.~C. Sacher}}
% First names are abbreviated in the running head.
% If there are more than two authors, 'et al.' is used.
%
\longversion{\institute{BTU, China \email{zhidanfeng@bjut.edu.cn} \and Trier University, Germany \email{\{fernau,kindermann,sacher\}@uni-trier.de} \and Kiel University, Germany \email{fpa@informatik.uni-kiel.de}}}
%

\begin{comment}
\author{Zhidan Feng}
{Universit\"at Trier, Fachbereich IV, Informatikwissenschaften, Germany \and Shandong University, School of Mathematics and Statistics, Weihai, China}
{zhidanfeng@mail.sdu.edu.cn}
{THTps://orcid.org/0000-0002-3364-5396}
{}

\author{Henning Fernau}
{Universit\"at Trier, Fachbereich IV, Informatikwissenschaften, Germany \and \url{THTps://www.uni-trier.de/index.php?id=49861}}
{fernau@uni-trier.de}
{THTps://orcid.org/0000-0002-4444-3220}
{}

\author{Pamela Fleischmann}
{Kiel University, Germany}
{fpa@informatik.uni-kiel.de}
{THTps://orcid.org/0000-0002-1531-7970} 
{}

\author{Kevin Mann}
{Universit\"at Trier, Fachbereich IV, Informatikwissenschaften, Germany}
{mann@uni-trier.de}
{THTps://orcid.org/0000-0002-0880-2513} 
{}

\author{Silas Cato Sacher}
{Universit\"at Trier, Fachbereich IV, Informatikwissenschaften, Germany}
{sacher@informatik.uni-trier.de}
{THTps://orcid.org/0009-0004-6850-1298} 
{}
\end{comment}

\maketitle              % typeset the header of the contribution

\begin{abstract}
The speed of a graph class~$\cG$ measures how many labeled graphs on~$n$ vertices one can find in~$\cG$. This graph class complexity function is explicitly provided on graphclasses.org. However, for many graph classes, their speed status is classified as \emph{unknown}. 
%In this paper, we classify quite a number of hitherto unknown graph classes as being factorial, meaning that its speed function behaves like $2^{\Theta(n\log n)}$. 
\longversion{In this paper, w}\shortversion{W}e show that any graph class representable by a finite binary language has at most factorial speed, meaning that its speed function behaves like $2^{\Theta(n\log n)}$, and we use this criterion to classify many graph classes whose speed was previously unknown as factorial. 
%To prove these results, we first show that each of these graph classes admits a representation by a finite language in the sense of generalized word representability. 
As a consequence, inclusions between several graph classes can now be seen to be proper. We also prove that $k$-letter graphs have exponential speed, i.e., the speed function lies in~$2^{\Theta(n)}$. 
\end{abstract}

\setcounter{footnote}{0}

\section{Introduction}

Many natural graph classes sit precisely at the boundary between factorial and superfactorial speed. While inclusion-based arguments often show at least factorial speed, proving an upper bound has remained elusive.
We study graph classes that can be described by finite languages and explain how a relatively simple reasoning about those graph classes suffices to settle the question if such a graph class has factorial speed.
This reasoning method is then applied to quite a number of classes whose speed is described as \emph{unknown} in \href{http://graphclasses.org/classes/speed.html}{graphclasses.org/classes/speed.html}~\cite{isgci} on December 1st, 2025, proving factorial speed for all of them. A survey on this situation %from the point of view of this paper 
is given in \myref{fig:Survey-graph-classes}. Whenever we refer to an \emph{unknown speed} in this paper, we mean exactly this situation.

\afterpage{
\begin{landscape}
\begin{figure}[p] 
% \centering
%\includegraphics[angle=90,width=\paperheight,height=\paperwidth,keepaspectratio,page=3]{IWOCA (factorial speed)/img/inclusion_diagram.pdf}
\vspace{-1cm}\hspace{-5cm}\resizebox{2\textwidth}{!}{\input{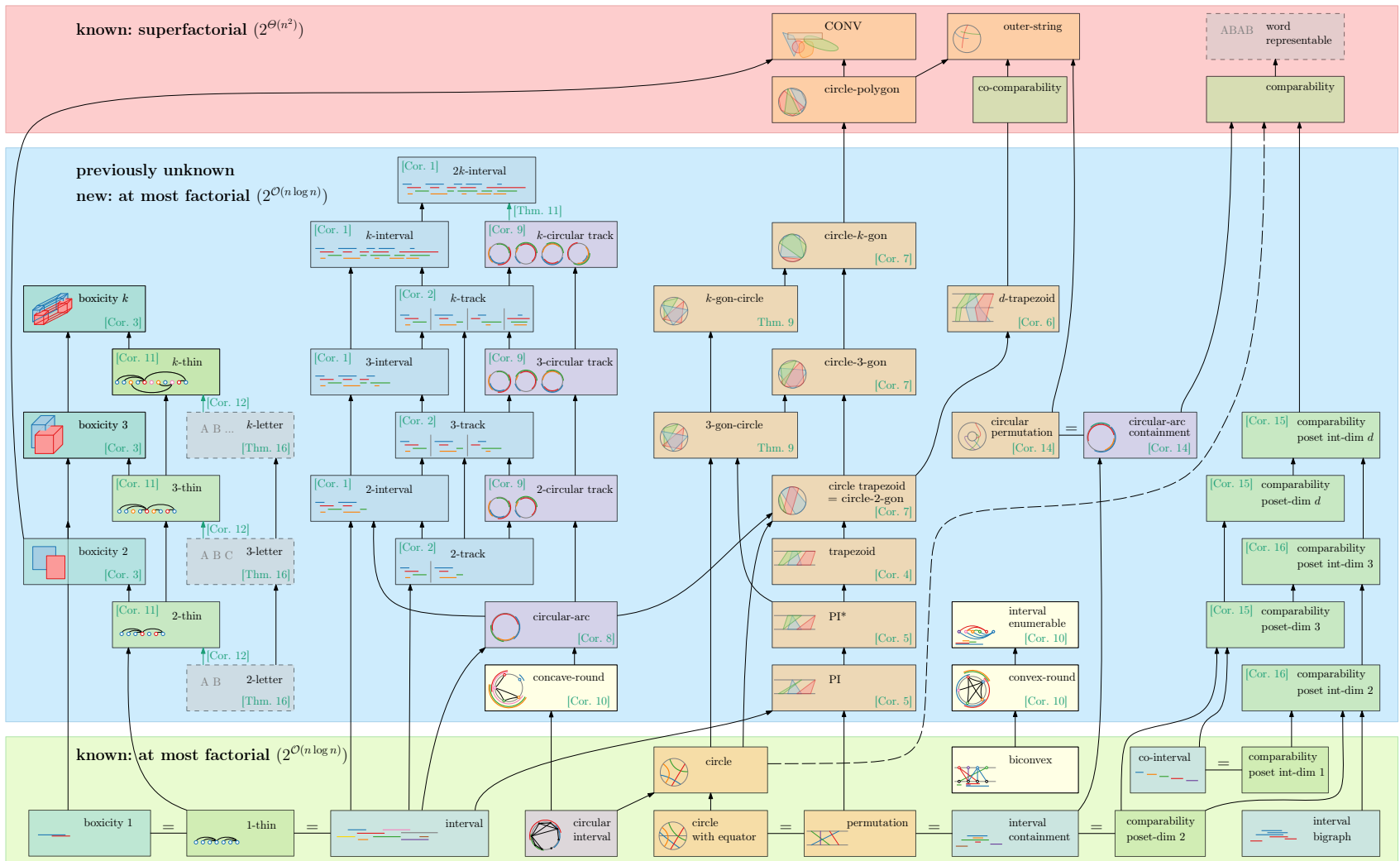}}

\caption{Summary of speed and inclusion results.}
\label{fig:Survey-graph-classes}
\end{figure}
\end{landscape}
}

The key observation underlying all our results is that representability by a finite binary language already implies a strong upper bound on the speed of a graph class. More precisely, if a graph class can be described via generalized word representability using a finite language, then its speed is at most factorial. This simple counting argument allows us to resolve many previously open cases in a uniform way.
Trivially, if $\cG\subseteq\cG'$ and \longversion{if we know that} $\cG'$ has at most factorial speed, then so has $\cG$, and if $\cG$ has at least factorial speed, then so has~$\cG'$. These observations, together with the known factorial speed of interval and permutation graphs, then implies that we can clarify the speed question for any graph class~$\cC$ marked as unknown 
 in \myref{fig:Survey-graph-classes} once we found a finite  language representation for~$\cC$ or a superclass thereof. We can see, as a second line of results, that we can use finite languages for many graph classes in the sense of generalized word representability as introduced in~\cite{FenFFMS2024}; this shows the versatility of this approach and could inspire further studies. Finally, by proving factorial speed for certain graph classes we can deduce so far open separation results from  superclasses with known superfactorial speed. Additionally, we prove that $k$-letter graphs have exponential speed and that they form a strict subclass of the $k$-thin graphs.

\section{Preliminaries}

We \longversion{have to }fix some notation\longversion{ before being able to formulate and prove our results, in particular,} as we are connecting different areas with our formalism.

\paragraph{General Notions.}
$\mathbb{N}$ denotes the set of non-negative integers. For $k\in \mathbb{N}$, $\mathbb{N}_{\geq k}=\{\ell\in\mathbb{N}\mid \ell\geq k\}$ and 
$[k] = \{1,2, \dots, k \}=\mathbb{N}_{\geq 1}\setminus\mathbb{N}_{\geq k+1}$. If $A$ is some set, then $\binom{A}{k}$ denotes the set of all $k$-element subsets of~$A$. We write $A\uplus B$ for disjoint union. 
Let $\intervals$ denote all closed intervals with real numbers as endpoints. If $I\in \intervals$, then $\lborder(I)$ and $\rborder(I)$ denotes the left and right endpoint of~$I$, respectively.

\paragraph{Formal Language Notions.}
An \emph{alphabet} $\Sigma$ is a non-empty, finite set with \emph{letters} as elements. A \emph{word} over $\Sigma$ is a finite \emph{concatenation} of letters from~$\Sigma$; $\Sigma^{\ast}$ denotes the set of all words over~$\Sigma$, including the \emph{empty word}~$\emptyword$. Any $L\subseteq\Sigma^*$ is called a \emph{language} over $\Sigma$. $\Sigma^{\ast}$ together with the concatenation operation~$\cdot$ forms a monoid, the so-called \emph{free monoid} of $\Sigma$, written $(\Sigma^*,\cdot)$. 
A mapping $f$ from $(\Sigma^{\ast},\cdot)$ into another monoid $(M,\circ)$ is called a \emph{morphism} if
$f(x\cdot y)=f(x)\circ f(y)$ holds for all $x,y\in\Sigma^{\ast}$. Note that a morphism is uniquely defined by giving the images of all letters. Mostly, the operator~$\cdot$ is suppressed.
%For $A \subseteq \Sigma$,  the \emph{projective morphism} $h_A: \Sigma^* \rightarrow A^*$ is defined by  $h_A(\ta) = \ta$ for $\ta \in A$ and $h_A(\ta) = \emptyword$, otherwise. Thus, we have $h_{\{\ta\}}(\mathtt{banana})=\ta\ta\ta$. 
We describe some morphisms important for this paper now.
For $\{\ta,\tb\}\in\binom{\Sigma}{2}$, let $h_{\ta,\tb}:\Sigma^*\to\{0,1\}^*$ be defined by $\ta\mapsto 0$, $\tb\mapsto 1$, $x\mapsto\emptyword$ for $x\in\Sigma\setminus\{\ta,\tb\}$. Let $\widetilde{\cdot}:\{0,1\}^*\rightarrow\{0,1\}^*$ be the \emph{complement morphism} mapping $0$ to $1$ and $1$ to $0$ (e.g., $\widetilde{010}=101$). %Note that the complement morphism is an involution. 
For $L\subseteq\{0,1\}^*$, let $\langle L\rangle=\{u,\widetilde{u}\mid u\in L\}$; $L$ is \emph{$0$-$1$-symmetric} if $L=\langle L\rangle$.
\longversion{ 

In addition to the concatenation of words, there exists the \emph{shuffle product}
to combine two words in a different way: 
for $u,v\in\Sigma^*$, define the \emph{shuffle} $u\shuffle v$ by 
\[
\{x_1y_1x_2y_2\cdots x_ny_n\mid \exists x_1,\dots,x_n, y_1,\dots,y_n\in\Sigma^*: u=x_1\cdots x_n\land v=y_1\cdots y_n\}.
\]
For example, $\td \ta \tb \te \tc = (\emptyword) (\td) (\ta \tb) (\te) (\tc) (\emptyword) \in \ta \tb \tc \shuffle \td \te$.}
The \emph{length of a word $w$} is denoted by $|w|$.
 For $w\in\Sigma^{\ast}$ and for all $i\in[|w|]$, $w[i]$ denotes the $\nth{i}$ letter of~$w$, i.e., $w = w[1] \cdots w[|w|]$. For $1\leq i\leq j\leq |w|$, $w[i, j]=w[i]\cdots w[j]$ is a \emph{factor} of~$w$. $\Sigma^n\coloneqq\{w\in\Sigma^*\mid |w|=n\}$\longversion{ collects all words of length~$n$}. 
The \emph{number of occurrences} of $\ta\in\Sigma$ in $w\in\Sigma^*$ is defined as $|w|_{\ta}\coloneqq|\{i\in[|w|]\mid w[i]=\ta\}|$ and $w$'s alphabet is given by $\alphabet(w)\coloneqq\{\ta\in\Sigma\mid\exists i\in[|w|]:\,w[i]=\ta\}$ as the symbols that occur in~$w$.
A word $w\in\Sigma^{\ast}$ is called \emph{$k$-uniform} for some $k\in\N$ if $|w|_{\ta}=k$ for all $\ta\in\Sigma$. 
%A word $w\in\Sigma^*$, with $\Sigma=\alphabet(w)$, is called \emph{$k$-uniform} if each letter occurs exactly $k$ times. 
Let $\Sigma^{k\text{-uni}}$ be the language containing all $k$-uniform words over~$\Sigma$. \shortversion{For $u,v\in\Sigma^*$, 
\[
\{x_1y_1x_2y_2\cdots x_ny_n\mid \exists x_1,\dots,x_n, y_1,\dots,y_n\in\Sigma^*: u=x_1\cdots x_n\land v=y_1\cdots y_n\}
\]}
defines the \emph{shuffle} $u\shuffle v$. 
If $\Sigma=\{\ta_1,\dots,\ta_n\}$, then $\Sigma^{k\text{-uni}}=\ta_1^k\shuffle\ta_2^k\shuffle\cdots\shuffle\ta_n^k\,.$ We also speak of the \emph{index~$\position_\ta(\ell,w)=j$ of the $\nth{\ell}$ occurrence} of a letter $\ta\in\Sigma$ in~$w$, referring to $j\in [k|\Sigma|]$ such that $w=u\ta v$,  $|u|=j-1$ and $|u|_{\ta}=\ell-1$.
%For $i\in [d]$, let $\text{pos}_\ta(i,w)$ select the position of the $i$-th occurrence of letter~$\ta$ in the $d$-uniform word~$w$.
%where $u\in (\Sigma\setminus\{\ta\})^*\shuffle\ta^{\ell-1}$ and $v\in (\Sigma\setminus\{\ta\})^*\shuffle\ta^{k-\ell}$ and $j=|u\ta|$.
%\todofpa{the definition of the $\nth{\ell}$ occurrence seems too complicated to me: $w=uav$, $w[j]=a$, $|u|_{a}=\ell-1$; moreover by your definition $w$ has to be $k$-uniform}
For $w\in \Sigma^*$, let $\Sigma_k(w)=\{\ta\in\Sigma\mid |w|_\ta=k\}$.

\newcommand{\deletion}[3]{\delta_{#1,#2}^{#3}}
\begin{definition}
For all $\ell,d \geq 1$ and 
$k,m \in [\ell]$, let 
%$\deletion{k}{m}{2\ell}:\{0,1\}^{2\ell\text{-uni}}\to \{0,1\}^{2\text{-uni}}$ be the function that deletes any but the $\nth{(2k-1)}$ and the $\nth{(2k)}$ occurrence of~$0$ and any but the $\nth{(2m-1)}$ and the $\nth{(2m)}$ occurrence of $1$ from~$w$.
$\deletion{k}{m}{2\ell\to 2d}:\{0,1\}^{2\ell\text{-uni}}\to \{0,1\}^{2d\text{-uni}}$ be the function deleting any but the $\nst{(2k-1)}$ up to the $\nth{(2(k-1+d))}$ occurrence of~$0$ and any but the $\nth{(2m-1)}$ up to the $\nth{(2(m-1+d))}$ occurrence of $1$ from~$w$. If $d=1$, we omit writing $\to 2$ in the exponent.
\end{definition}

Intuitively, \longversion{in our applications, }the deletion operator extracts the relative order of a fixed pair of intervals (or dimensions)\shortversion{,}\longversion{ while} discarding all others.
\longversion{For example, for}\shortversion{Consider, e.g.,} $w = 001001110101$, \shortversion{then }\longversion{we find:\begin{itemize}\item}$\deletion{2}{1}{6}(w) = \deletion{2}{1}{6}({\color{gray}00}1001{\color{gray}110101}) = 1001$,\longversion{\item} $\deletion{3}{2}{6}(w) = \deletion{3}{2}{6}({\color{gray}001001}110{\color{gray}1}0{\color{gray}1}) = 1100 $ and\longversion{\item} $\deletion{2}{1}{6\to 4}(w) = \deletion{2}{1}{6\to 4}({\color{gray}00}1001110{\color{gray}1}0{\color{gray}1}) =10011100$.\longversion{\end{itemize}

Another type of deletion operator that we sometimes use is $d_\ta:\Sigma^*\to\Sigma^*$, for $\ta\in\Sigma$, that maps $w\in\Sigma^*$ either to $w$ (if $|w|_\ta=0$) or to the word~$u$ obtained from $w$ by deleting the first occurrence of~$\ta$ from~$w$, i.e., $|u|=|w|-1$.}

\paragraph{Graphs.} We only consider finite, undirected, simple graphs. These can be specified as $G=(V,E)$, where $V(G) = V$ is a finite set of \emph{vertices} and $E(G) = E \subseteq \binom{V}{2}$ is a set of \emph{edges}. If $A\subseteq V$, then $G[A]=(A,\{e\in E\mid e\subseteq A\})$ is the graph \emph{induced} by $A$. The \emph{complement} of~$G$ is the graph $\overline{G}=(V,\binom{V}{2}\setminus E)$.
If $G_1=(V_1,E_1)$ and $G_2=(V_2,E_2)$ are graphs, then a bijection $\varphi:V_1\to V_2$ 
is a \emph{graph isomorphism} \iffl $\{u,v\}\in E_1\iff \{\varphi(u),\varphi(v)\}\in E_2$.
Graphs are mostly identified when they are isomorphic; isomorphism is made explicit by writing~$\cong$. Abstract graphs in this sense are collected in graph classes~$\cC$, but these may contain concrete graphs, so that we can write $G\in\cC$. The class $\text{co-}\cC$ collects all complements of graphs from~$\cC$. 
A graph class is \emph{hereditary} if it is closed under taking induced subgraphs. 
Let $\cC$ be a graph class and $\cC^n = \{ G \in \mathcal{C} \mid V(G) = [n]\}$ for each $n \in \mathbb{N}\setminus\{0\}$. The function $n \mapsto |\cC^n|$ is the \emph{speed} of $\cC$. A graph class $\cC$ is called \emph{factorial} if $n^{c_1 n} \leq |\cC^n| \leq n^{c_2 n}$ for some constants $c_1,c_2 > 0$. Accordingly, $\cC$ is called \emph{at most factorial} if a constant $c > 0$ exists such that $|\mathcal{C}^n| \leq n^{c n}$, and \emph{superfactorial} if it is not at most factorial. Clearly, $\cC$ is factorial \iffl  $\text{co-}\cC$ is factorial. The speed of~$\cC$ is \emph{exponential} if there are $c_1,c_2 > 0$ such that  $c_1^n \leq |\cC^n| \leq c_2^n$.

\paragraph{Known Graph Classes and Their Representations.} (See \myref{fig:Survey-graph-classes}.)
We mostly study the following graph classes that can be described as intersection models:\shortversion{\\}
\longversion{\begin{itemize}
    \item} \underline{$\ell$-interval graphs:} each vertex is represented by a collection of $\ell$ closed intervals of the real line. They were originally introduced in~\cite{TroHar79}.  1-interval graphs are interval graphs, \longversion{every outerplanar graph is a 2-interval graph~\cite{KosWes99}, }and every planar graph is a 3-interval graph~\cite{SchWes83}.
    \longversion{\item}\shortversion{\\} \underline{$\ell$-track graphs}, introduced as multidimensional interval graphs in \cite{KumDeo94}: the intersection model is based on tracks, i.e., parallel lines, such that each of the $\ell$ intervals $I_i(x)$ that can be associated to a vertex~$x$ lies on a specific track. \longversion{We can think of an edge between $u$ and $v$ existing \iffl  the first intervals of $u$ and $v$ intersect, or the second ones intersect, or \dots, or the $\nth{\ell}$ intervals of $u$ and $v$ intersect. }\shortversion{There is an edge between $u$ and $v$ iff $I_i(u)\cap I_i(v)\neq\emptyset$.} 1-track graphs are the interval graphs. $\ell$-track graphs are $\ell$-interval graphs\longversion{, but not vice versa if $\ell>1$}.\longversion{\footnote{It should be noted that the \emph{track layouts} introduced in \cite{DujPorWoo2004}, as well as related graph parameters, are something completely different.}}
     \longversion{\item}\shortversion{\\} \underline{boxicity $b$} graphs: their intersection model is based on $b$-dimensional axis-parallel boxes. This concept has been introduced by Roberts~\cite{Rob69}. As each such box is given by $b$ intervals $I_i$, we can think of an edge between $u$ and~$v$ existing \iffl \longversion{the first intervals of $u$ and~$v$ intersect, and the second ones intersect, and \dots, and the $\nth{b}$ intervals of $u$ and $v$ intersect}\shortversion{$I_i(u)\cap I_i(v)\neq\emptyset$ for all $1\leq i\leq b$}. For $b=1$, we face again the interval graphs.\longversion{ The special class of graphs of boxicity at most~$2$ is also known as  rectangle intersection graphs, see~\cite{QueWeg90}.}
 \longversion{\item}\shortversion{\\} \underline{$d$-trapezoid graphs:} their intersection model consists of $d$-trapezoids between $d$ parallel lines. More precisely, we associate an interval $[\ell_i,r_i]$ to line~$i$. We obtain the trapezoid by connecting $\ell_i$ with $\ell_{i+1}$ for $i\in [d-1]$, $\ell_d$ with $r_d$, $r_{i+1}$ with $r_i$  for $i\in [d-1]$ and finally $r_1$ with $\ell_1$. This model has been introduced in~\cite{Flo95} (with a dimension count off by one), also confer~\cite{FelMulWer97}. $2$-trapezoid graphs with all involved intervals of length~1 characterize permutation graphs~\cite{BodKKM98}.
%Let L1 ...Ld be d parallel lines in the plane. A d-trapezoid is the polygon obtained by choosing an interval Ii on every line Li and connecting the left, resp. right endpoint of Ii with the left, resp. right endpoint of Ii+1. A graph is a d-trapezoid graph if it has an intersection model consisting of d-trapezoids between d parallel lines.
\longversion{\item \underline{$k$-gon-circle graphs:} their intersection model consists of circle polygons, i.e., convex polygons inscribed in a circle, with at most $k$ corners lying on the circle line, see \cite{JanKra92,KosKra97}. The recognition problem even for 3-gon-circle graphs (and also for polygon-circle graphs, not bounding the number of corners) is \NP-complete, see \cite{KraPer2003,Per2007,Per2008}.
\item
}\shortversion{\\}
\underline{circle-$k$-gon graphs:} introduced  in \cite{Gav2011} to generalize $k$-gon-circle graphs, see \cite{JanKra92,KosKra97}. \longversion{More precisely, a}\shortversion{A} circle-$k$-gon is the region between $k$ or fewer non-crossing chords of a circle, no chord
connecting the arcs between two other chords; the sides of a circle-$k$-gon are alternatingly
chords or arcs of the circle. A \emph{circle-$k$-gon graph} is the intersection graph of a family of
circle-$k$-gons in a circle.\longversion{ The family of circle trapezoid graphs is exactly the family of
circle 2-gon graphs and the family of circle graphs is exactly the family of circle 1-gon
graphs.} 
\longversion{\end{itemize}}

More details on intersection models generalizing interval graphs can be found in~\cite{Sch84}. For all these graph classes, it does not really matter whether we consider closed intervals or open intervals, or intervals with integer boundaries.
%We will see quite a number of results concerning trapezoid and circular-arc graphs and variations thereof.\todo{reformulate; I think we have to be very sketchy for the circular-arc variations}

Another way of looking at and generalizing interval graphs was suggested in~\cite{ManORC2007}. A graph $G=(V,E)$ is \emph{$k$-thin} if its vertices admit a linear ordering $v_1<v_2<\cdots<v_n$ and a partition $V=V^1\uplus V^2\uplus\cdots\uplus V^k$ such that for every $a<b<c$ with $v_a,v_b\in V^i$ and $\{v_a,v_c\}\in E$, we also have $\{v_b,v_c\}\in E$.

%\todohf{%Possibly, \myref{rem:interrelations} is helpful for the picture. 
%Also, \myref{rem:boxicity-versus-dimension} discusses inclusions between boxicity and dimension hierarchies.}

Yet another way to define graph classes is via (co-)comparability graphs of certain partial orders. 
The \emph{dimension} of a partially ordered set~$V$ is the minimum number of linear orders on~$V$ whose intersection is the partial order on~$V$. This notion was introduced in~\cite{DusMil41}.
%According to graphclasses.org, e
Even for dimension $d=3$, the class of comparability graphs of posets of dimension~$d$ has unknown speed. \longversion{This notion of dimension corresponds to the uniformity of words that permutationally represent graphs; see \cite{KitSei2008,KitLoz2015}.}
Interestingly, for $d=3$, \href{http://graphclasses.org}{graphclasses.org} lists 12 maximal subclasses %of graphs all of which have known 
with factorial speed.

The \emph{interval dimension} of a poset was introduced in~\cite{TroBog76}.
 A partial order is an \emph{interval order} if we can associate closed intervals~$I_x$ of the real line to elements~$x$ and set $x\prec y$ \iffl $\forall \xi\in I_x\forall \eta\in I_y: \xi\leq \eta$, i.e., $I_x$ is completely to the left of $I_y$, possibly sharing at most one element.
As singletons can be viewed as closed sets, any linear order is an interval order.
The interval dimension of a poset is the smallest number of interval orders whose intersection equals the given poset. Again, one can define comparability graphs of posets of bounded interval dimension.
Comparability graphs of posets of bounded interval dimension~1 are exactly the co-interval graphs.
As comparability graphs of posets of bounded interval dimension~2 are exactly the co-trapezoid graphs, so we can handle them.

A graph $G=([n],E)$ is a \emph{$k$-letter graph} if there exists a word $w\in [k]^n$ and a \emph{decoder}  $\cD\subseteq [k]^2$ such that, for $1\leq i<j\leq n$, $\{i,j\}\in E$ \iffl $w[i]w[j]\in \cD$. We also write $G=G(\cD,w)$. 
We will encounter further graph classes in passing whose definitions can be found \longversion{when they are used}\shortversion{in the appendix, where also many proofs have been moved (marked by a clickable `$(*)$')}.

\paragraph{Interpreting Formal Languages as Graph Classes.}
We are now ready to formulate our main link between languages and graph classes, showing how we can interpret words over some alphabet~$V$ as graphs with the help of a language $L\subseteq\{0,1\}^*$.

\begin{definition}
Let $L\subseteq\{0,1\}^*$ be $0$-$1$-symmetric. Let $w\in V^*$ with $V=\alphabet(w)$.
We call a graph $G=(V,E)$ \emph{$L$-represented by $w\in V^{\ast}$} iff, for all $u,v\in V$ with $u \neq v$, we have
\(\{u,v\}\in E \Leftrightarrow h_{u,v}(w)\in L. \)
A graph  $G=(V,E)$  is \emph{$L$-representable} if it can be $L$-represented by some word $w\in V^*$. 
Let $G(L,w)$ denote the graph $G$ that is $L$-represented by $w\in\Sigma^{\ast}$ and let $\cG_L=\{G(L,w)\mid w\in\Sigma^*, \Sigma=\alphabet(w)\}$.
\end{definition}

As $L\subseteq\{0,1\}^*$ is $0$-$1$-symmetric, $h_{u,v}(w)\in L$ \iffl $h_{v,u}(w)\in L$. Hence, for each $0$-$1$-symmetric  language $L\subseteq\{0,1\}^*$, the graph class $\cG_L$ is well-defined. Also, for any $G=G(L,w)$ and any $A\subseteq \alphabet(w)=V(G)$,  $G[A]\in \cG_L$ as $G[A]=G(L,h_A(w))$, where $h_A$ acts as the identity on~$A$ and maps other letters to~$\emptyword$. In other words, any $\cG_L$ is hereditary. This is interesting, as several papers on the speed of graph classes focus on hereditary classes; see \cite{BalBSS2009,BalBolWei2000,BalBolWei2001,LozMayZam2011,SchZit94}.

%\todohf{It might be an idea to explain this approach with some examples. Natural candidates would be: classical word-representable graphs, interval graphs, permutation graphs, circle graphs.}
\shortversion{\noindent}We will now explain this approach with some examples.

\begin{example}\label{exa:word-representability}
One interesting way of representing graphs was introduced by Kitaev and Pyatkin \cite{KitPya2008}: a graph is represented by a word $w$ that has the graph's vertices as letters and an edge $\{\ta,\tb\}$ is given \iffl $\ta$ and $\tb$ alternate in~$w$.
For instance, the graph $G=(\{\ta,\tb,\tc\},\{\{\ta,\tb\}\})$ is represented by the word 
$\ta\tb\tc\tc\ta$. \shortversion{We refer to the monography~\cite{KitLoz2015}.}\longversion{The class of word-representable graphs also found some interest in the pure graph-theoretic literature, see, e.g.,  \cite{CheKitSun2016,ChoKimKim2019,ColKitLoz2017,EnrKit2019,GleKitPya2018,Gle2019,KitLoz2015,KitSai2020,SriHar2024,SriHar2025}.} 
\emph{Classical} word representability based on alternation is modeled in our setting by the binary  $0$-$1$-symmetric language 
$L_{\classical}=\langle\{\emptyword,1\}\{01\}^*\{\emptyword,0\}\rangle,$ leading to the graph class $\cG_{L_{\classical}}$. 
\longversion{Kitaev and Lozin showed in~\cite[Theorem 5.1.7]{KitLoz2015} that}\shortversion{As shown in \cite[Theorem 5.1.7]{KitLoz2015},} the sublanguage of 2-uniform alternating words $\langle 0101\rangle$ already describes an interesting subclass of word-representable graphs\shortversion{:}\longversion{. More precisely,} 
$\cG_{\langle 0101\rangle}$ is the class of circle (or \longversion{interval }overlap) graphs. 
\end{example}

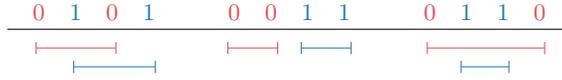
\begin{figure}[tb]
    \centering
    \begin{tikzpicture}[scale=0.5]
\draw[] (0,0) -- (15,0) node[pos=0.5, sloped, above] {{\color{DR}0}\quad{\color{DB}1}\quad{\color{DR}0}\quad{\color{DB}1}\quad\quad\quad
{\color{DR}0}\quad{\color{DR}0}\quad{\color{DB}1}\quad{\color{DB}1}\quad\quad\quad
{\color{DR}0}\quad{\color{DB}1}\quad{\color{DB}1}\quad{\color{DR}0}
};
\draw[|-|,DR] (0.75,-0.5) -- (2.9,-0.5) node[pos=0.5, sloped, above] {};
\draw[|-|,DB] (1.75,-1) -- (3.95,-1) node[pos=0.5, sloped, above] {};
\draw[|-|,DR] (5.85,-0.5) -- (7.2,-0.5) node[pos=0.5, sloped, below] {};
\draw[|-|,DB] (7.8,-0.5) -- (9.15,-0.5) node[pos=0.5, sloped, below] {};
\draw[|-|,DR] (11.15,-0.5) -- (14.3,-0.5) node[pos=0.5, sloped, above] {}; 
\draw[|-|,DB] (12.05,-1) -- (13.35,-1) node[pos=0.5, sloped, below] {};

\end{tikzpicture}
    \caption{Three interval relation patterns: overlap, disjointness and containment.}
    \label{fig:l-example}
\end{figure}

\begin{example}
Further\longversion{ well-known} graph classes can be described by 2-uniform languages:
\longversion{\begin{itemize}
    \item $\cG_{\langle 0101,0110\rangle}$ is the class of interval graphs;
    \item $\cG_{\langle 0110\rangle}$ is the class of permutation (or interval containment) graphs;
    \item $\cG_{\langle 0011\rangle}$ is the class of co-interval graphs.
\end{itemize}}
\shortversion{(a) $\cG_{\langle 0101,0110\rangle}$ is the class of interval graphs; (b) $\cG_{\langle 0110\rangle}$ is the class of permutation (or interval containment) graphs; (c) $\cG_{\langle 0011\rangle}$ is the class of co-interval graphs.} 
For a complete classification of graph classes described by 2-uniform languages, we refer to \cite{FenFFMS2024}. For this paper, the following two thoughts are important.
(1) For instance, looking at the two words $0101,0110$ that are meant to describe interval graphs, one can imagine the two situations in which intervals may interact in order to create an edge in the well-known interval graph model: either they overlap ($0101$) or one is contained in the other one ($0110$). For illustration, we refer to \myref{fig:l-example}. \longversion{This also explains why circle graphs, also known as overlap graphs, are described by $0101$ and why permutation graphs, also known as interval containment graphs, are described by $0110$. }(2) All these graph classes are known to possess factorial speed\longversion{, see graphclasses.org}. 
\end{example}

\section{Main Tool Proving At Most Factorial Speed}

The following theorem formulates the main combinatorial result. Despite its simplicity, it has many applications and helps solve open questions. % from graphclasses.org.

\begin{theorem}\label{thm:counting-Lfinite}
If $L\subseteq\{0,1\}^*$ is finite, then $\cG_{\langle L\rangle}$ has \longversion{speed at most $2^{O(n\log n)}$ and hence is at most factorial}\shortversion{at most factorial speed}.
\end{theorem}

\begin{proof} We are going to give an information-theoretic argument.
If $L$ is finite, then so is $\langle L\rangle$. Assume $L=\langle L\rangle$. For each $G=(V,E)\in \cG_L$, there is a word $w\in V^*$ of length at most $(\text{ml}(L) + 1)\cdot |V|$ (where $\text{ml}(L)$ denotes the maximum length of any word in~$L$) with $G(L,w)=G$. Namely, if we had to consider a graph $G(L,u)$ for some $u$ of length more than  $(\text{ml}(L) + 1)\cdot |V|$, then by the pigeon-hole principle, there will be a vertex~$\ta$ that occurs more than $\text{ml}(L)+1$ many times. Hence, $\ta$ is an isolated vertex in $G(L,u)$, and there is a word $u'$ that can be obtained from~$u$ by deleting all but the first  $\text{ml}(L)+1$ many occurrences of $\ta$, while $G(L,u)=G(L,u')$.
By iterating this argument, one can see the claimed upper bound on the length of words that need to be considered in order to describe each graph of $\cG_L$. As  $\text{ml}(L)$ is a constant only depending on the fixed language~$L$, writing down $w\in V^*$ in order to describe $G\in \cG_L$ costs at most $\cO(|V|\log(|V|))$ many bits.
Hence, at most $2^{\cO(|V|\log(|V|))}$ many graphs with $|V|$ many vertices can belong to $\cG_L$. 
\qed 
\end{proof}

This immediately implies that no graph class with superfactorial speed can be described by any finite language.  For instance, this means that no finite language can represent the class of word-representable graphs. In the following, we apply \myref{thm:counting-Lfinite} %this theorem 
by giving representations of graph classes by finite languages, hence proving that their speed can be at most factorial. Although this theorem has no complicated proof, it is quite powerful as a uniform tool to resolve hitherto unknown speed questions for certain graph classes.

\begin{toappendix}
\begin{remark}
Let us continue our example concerning classical word representability. Recall that any word-representable graph~$G$ can be represented by some $k$-uniform word. The smallest possible $k$ such that there exists a $k$-uniform word~$w$ such that $G\cong G(L_{\classical},w)$ is also known as the \emph{representation number} of~$G$.
In \cite{DwaMozKri2025}, the authors showed that the representation number of word-representable split graphs is at most~3. Using our formalism, this means that word-representable split graphs belong to $\cG_L$ with $L=\langle 01(01 \cup\emptyword)(01\cup\emptyword)\rangle$. As both the class of split graphs and the class of word-representable graphs are superfactorial, while $\cG_L$ is at most factorial by \myref{thm:counting-Lfinite}, the class of split graphs and the class of word-representable graphs are incomparable. More on word-representable split graphs can be found in \cite{CheKitSai2022,Iam2022}.
\end{remark}

\begin{remark}\label{rem:adding-isolates}
It is not hard to see that for any finite language $L\subseteq\{0,1\}^*$, $\cG_{\langle L\rangle}$ is closed under adding isolated vertices: any letter~$\ta\in \alphabet(w)=V$ of a word $w\in V^*$ with $|w|_\ta\notin \{|u|_0\mid u\in\langle L\rangle\}$ encodes an isolated vertex in $G(\langle L\rangle,w)$. Therefore, when we are considering a graph class $\cG$ with our methodology such that $\cG$ is not (known to be) closed under adding isolates, then we are actually showing that the superclass $\cG'$ obtained from $\cG$ by adding to each graph $G\in\cG$ an arbitrary number of isolated vertices can be represented by a finite language and hence has at most factorial speed. But then, also $\cG$ has at most factorial speed.
\end{remark}
\end{toappendix}

\section{Applications: New Speed Results}

In this section, we collect a number of consequences of our main result that help resolve several cases marked as \textit{unknown} in graphclasses.org concerning the question whether a certain graph class has factorial speed.
\longversion{All results in this section follow the same pattern: we exhibit a finite uniform language encoding the graph class and then invoke \myref{thm:counting-Lfinite}. In some cases, speed results are also derived by inclusions of graph classes.}

%\todo[inline]{HF would suggest to include proofs for characterizations of graph classes only in the main part if the speed results do not follow by inclusion. \\ Silas: agreed}

\longversion{\subsection{Multi-Interval and Multi-Track Graphs}}
\shortversion{\subsection{Multi-Interval, Multi-Track, and Bounded Boxicity Graphs}}

%The $\ell$-interval graphs are the graphs with an intersection model consisting of $\ell$ (not necessarily disjoint) closed intervals on the real line, originally introduced in~\cite{TroHar79}. 
%Clearly, for each $\ell$-interval graph, we can find such an $\ell$-interval representation with pairwise disjoint intervals and distinct endpoints.

We will start our discussion of graph classes defined by geometric intersection models by considering multi-interval graphs. We \longversion{will }discuss them to some details\longversion{ as several other models can be treated alike}. 

\begin{theorem}\label{thm:lIntervalGraphs}
For each $\ell\in\mathbb{N}_{\geq 1}$, $\cG_{L_\ell^{\text{int}}}$ is the class of  $\ell$-interval graphs, where 
%and the finite language 
$L_{\ell}^{\text{int}} \coloneqq \{ w \in \{0,1\}^{2\ell\text{-uni}}
%0^{2\ell} \shuffle 1^{2\ell} 
\mid \exists k,m \in [\ell]: \deletion{k}{m}{2\ell}(w) \in \langle 0101, 0110 \rangle\}$ is $2\ell$-uniform.
\end{theorem}

\begin{proof}By construction, $L_\ell^{\text{int}}$ is $2\ell$-uniform.
%\begin{enumerate}
%\item 

Let $G = (V,E)$ be an  $\ell$-interval graph with representation
$f_G: V \to \binom{\intervals}{\ell}$. W.l.o.g., assume that the intervals of $f_G(v)$ are pairwise disjoint for all $v\in V$ and that all the endpoints~$x_i$ are distinct\shortversion{, ordered as $x_1 < \dots < x_{2\ell |V|}$}. \longversion{Let $x_1 < \dots < x_{2\ell |V|}$ be the endpoints of the intervals in the $\ell$-interval representation of~$G$. }Define $w\in V^{2\ell |V|}$ by $w[i] \coloneqq v$ if $x_i$ is an endpoint of an interval in $f_G(v)$ for all $i \in [2 \ell |V|]$. 
%Let $w  \coloneqq \prod_{i \in [2 \ell |V|]}w[i]$. 
Clearly, $G(L_{\ell}^{\text{int}},w) = G$. 

%\item 
For the other direction, let $G=(V,E) \in \mathcal{G}_{L_\ell^{\text{int}}}$ be represented by a word $w\in V^*$.  
Let $V_{2\ell} = \{ v \in V \mid |w|_v = 2\ell\}$. 
As $L_\ell^{\text{int}}$ is $2\ell$-uniform, all $v \in \overline{V_{2\ell}(w)}$ are isolated vertices. Let $w'$ be obtained from $w$ by deleting all letters of $ \overline{V_{2\ell}(w)}$ from~$w$. By construction, $w'$ is $2\ell$-uniform.
%, i.e., $w'=\ta_1\ta_2\cdots \ta_{2\ell n'}$, with $n'=|V_{2\ell}(w)|=|\alphabet(w')|$. 
Let $m \in [\ell]$. For $v \in V_{2\ell}(w)$, define $I_{m,v} \coloneqq [\position_v(2m-1,w'),\position_v(2m,w')]$.
%[e_1,e_2]$, where $e_1$ is the index of the $\nth{(2m-1)}$ and  $e_2$ is the index of the $\nth{(2m)}$ occurrence of $v$ in~$w'$. %Let $M \coloneqq \max \{ e_1,e_2 \mid v \in V_{2\ell} \wedge m \in [\ell]\}$. 
For all $v \in \overline{V_{2\ell}(w)}$, choose intervals $I_{m,v}$ such that %all the endpoints are greater than 
$|w'|<\lborder(I_{m,v})$ and all these intervals are pairwise disjoint. 
Let $f(v) \coloneqq \bigcup \{I_{m,v} \mid m \in [\ell]\}$. We want to show that $f$ is an $\ell$-interval representation of~$G$.
Note that $f(u) \cap f(v) \neq \emptyset$ \iffl there exists an interval $[r_1,r_2]\in \{I_{m,u} \mid m \in [\ell]\}$ that intersects with an interval $[s_1,s_2]\in \{I_{m,v} \mid m \in [\ell]\}$. As all interval endpoints are different, this means \shortversion{either overlap or containment}\longversion{ that either (1) $r_1<s_1<r_2<s_2$ or $s_1<r_1<s_2<r_2$ (overlap situations), or (2) $r_1<s_1<s_2<r_2$ or $s_1<r_1<r_2<s_2$ (containment situations)}. 

%Assume that (in the natural linear ordering of disjoint intervals, from left to right) $[r_1,r_2]$ is the $\nth{k_u}$ interval in $f(u)$ and that $[s_1,s_2]$ is the $k_v^{\mbox{\tiny th}}$ interval in $f(v)$. As $[r_1,r_2]\cap [s_1,s_2]\neq\emptyset$, these interval endpoints refer to positions of letters in~$w'$. More precisely, for $q\in[2]$,$r_q=\position_u(2k_u+(q-2),w')$ and $s_q=\position_v(2k_v+(q-2),w')$. 

%Consider $h_{u,v}(w')$. By construction, this word is $2\ell$-uniform. For $q\in[2]$, let $r_q'=\position_0(2k_u+(q-2),h_{u,v}(w'))$ and let $s_q'=\position_1(2k_v+(q-2),h_{u,v}(w'))$.
%By definition, their sequence determines the situations of the intervals as described above. More precisely, $r_1<s_1<r_2<s_2$ \iffl $r_1'<s_1'<r_2'<s_2'$, while $r_1<s_1<s_2<r_2$ \iffl $r_1'<s_1'<s_2'<r_2'$, etc., so that $\deletion{k_u}{k_v}{2\ell}(h_{u,v}(w'))\in\langle 0101\rangle$ describes an overlap situation, while $\deletion{k_u}{k_v}{2\ell}(h_{u,v}(w'))\in\langle 0110\rangle$ describes a containment. Hence, if $[r_1,r_2]\cap [s_1,s_2]\neq\emptyset$, then $\deletion{k_u}{k_v}{2\ell} (h_{u,v}(w'))\in\langle 0101,0110\rangle$. In that case, also $\deletion{k_u}{k_v}{2\ell}(h_{u,v}(w))\in\langle 0101,0110\rangle$.
%Conversely, for $u,v\in V_{2\ell}(w)$, if $\deletion{k_u}{k_v}{2\ell} (h_{u,v}(w'))\in\langle 0101,0110\rangle$, then $[r_1',r_2']\cap [s_1',s_2']\neq\emptyset$, implying $[r_1,r_2]\cap [s_1,s_2]\neq\emptyset$. 
%By construction, with $u\notin  V_{2\ell}(w)$, $u$ is an isolated vertex in~$G$, and so $f(u)$ is disjoint with $f(v)$ for any vertex $v\in V$.

Let $u,v \in V$. If $u \in \overline{V_{2l}}$ or $v \in \overline{V_{2l}}$, $f(u) \cap f(v) \neq \emptyset$ by construction and $h_{u,v}(w) \not\in L_{\ell}^{\text{int}}$ because $h_{u,v}(w)$ is not $2\ell$-uniform. Now assume that $u,v \in V_{2\ell}$. Hence, $h_{u,v}(w) = h_{u,v}(w')$. 
{\allowdisplaybreaks
\begin{align*}
& h_{u,v}(w) = h_{u,v}(w') \in L_{\ell}^{\text{int}} \\
\Leftrightarrow &  \exists k,m \in [\ell]: \delta_{k,m}^{2\ell}(h_{u,v}(w')) \in \{ 0101, 1010, 0110, 1001 \} \\
\Leftrightarrow & \exists k,m \in [\ell]: \\
& \phantom{\vee } \; \; \position_u(2k-1,w') < \position_v(2m-1,w') < \position_u(2k,w') < \position_v(2m,w') \\
& \vee \position_v(2m-1,w') < \position_u(2k-1,w') < \position_v(2m,w') < \position_u(2k,w') \\
& \vee \position_u(2k-1,w') < \position_v(2m-1,w') < \position_v(2m,w') < \position_u(2k,w') \\
& \vee \position_v(2m-1,w') < \position_u(2k-1,w') < \position_u(2k,w') < \position_v(2m,w') \\
\longversion{\Leftrightarrow & \exists k,m \in [\ell]: \\
& \phantom{\vee } \; \; [\position_u(2k-1,w'),\position_u(2k,w')] \cap [\position_v(2m-1,w'),\position_v(2m,w')] \neq \emptyset \\}
\Leftrightarrow & \exists k,m \in [\ell]: I_{k,u} \cap I_{m,v} \neq \emptyset \longversion{\\}
\Leftrightarrow \longversion{&} f(u) \cap f(v) \neq \emptyset
\end{align*}}

T\longversion{his concludes the proof that $f$ is an interval representation of $G$ and t}herefore, $G$ is an $\ell$-interval graph.
\qed 
%\end{enumerate}
\end{proof}

\begin{figure}[tb]
    \centering
    \begin{tikzpicture}[scale=0.5]
\draw[-] (-0.25,0) -- (15.25,0) node[pos=0.5, sloped, above] {${\color{DR}\ta}\quad{\color{DG}\tb}\quad{\color{DR}\ta}\quad{\color{DB}\tc}\quad
{\color{DG}\tb}\quad{\color{DG}\tb}\quad{\color{DG}\tb}\quad{\color{DO}\td}\quad
{\color{DB}\tc}\quad{\color{DR}\ta}\quad{\color{DO}\td}\quad{\color{DO}\td}\quad
{\color{DR}\ta}\quad{\color{DO}\td}\quad{\color{DB}\tc}\quad{\color{DB}\tc}$};
\draw[|-|,DR] (-0.25,1) -- (2.5,1) node[pos=0.5, sloped, above] {};
\draw[|-|,DB] (2.5,1) -- (8.5,1) node[pos=0.5, sloped, above] {};
\draw[|-|,DG] (0.5,1.5) -- (4.5,1.5) node[pos=0.5, sloped, below] {};
\draw[|-|,DO] (6.5,1.5) -- (10.5,1.5) node[pos=0.5, sloped, below] {};

\draw[-] (-0.25,-1.5) -- (15.25,-1.5) node[pos=0.5, sloped, above] {}; 
\draw[|-|,DG] (4.5,-0.5) -- (6.5,-0.5) node[pos=0.5, sloped, below] {};
\draw[|-|,DR] (8.5,-0.5) -- (12.5,-0.5) node[pos=0.5, sloped, below] {};
\draw[|-|,DB] (13.5,-0.5) -- (15.25,-0.5) node[pos=0.5, sloped, above]  {};
\draw[|-|,DO]  (10.5,-1) -- (13.5,-1) node[pos=0.5, sloped, below] {};

\node (qa) at (18,1) {{\color{DR}$\ta$}};
\node (qb) at (20,1) {\color{DG}$\tb$};
\node (qd) at (18,-1) {\color{DO}$\td$};
\node (qc) at (20,-1) {\color{DB}$\tc$};

\draw   (qa) edge[-] node{} (qb);
\draw   (qa) edge[-] node{} (qd);
\draw   (qc) edge[-] node{} (qb);
\draw   (qc) edge[-] node{} (qd);

\end{tikzpicture}
    \caption{Constructing a $2$-interval graph from a $4$-uniform word}
    \label{fig:l-interval-example}
\end{figure}
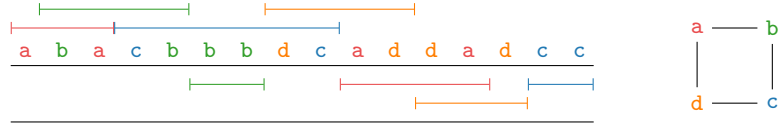
%\todoscs{replace the photograph by a tikzpicture}

\begin{example}
Let $w=\mathtt{abacbbbdcaddadcc}$ and $\ell=2$. Then, $G=G(L_2^{\text{int}})\cong C_4$, as $f_G(\ta)=\{[1,3],[10,13]\}$, $f_G(\tb)=\{[2,5],[6,7]\}$, $f_G(\tc)=\{[4,9],[15,16]\}$, $f_G(\td)=\{[8,11],[12,14]\}$\shortversion{; see \myref{fig:l-interval-example}}. Interestingly, $C_4$ is not an interval graph.\longversion{ \myref{fig:l-interval-example} is a sketch of this example.} 
\end{example}

\noindent
With \myref{thm:counting-Lfinite}, we can deduce:

\begin{corollary}\label{cor:interval}
The class of  $\ell$-interval graphs is factorial for each fixed $\ell\in\mathbb{N}_{\geq 1}$.
\end{corollary}

This has been known for (1-)interval graphs, but was considered unknown for larger interval numbers (as explicitly stated for $\ell=2$ and $\ell=3$). % in \url{https://graphclasses.org/classes/speed.html}. 
Yet, this result also follows with \cite[Lemma 1]{ErdWes85}. We presented the result including a detailed proof as many of the following constructions and proofs follow a similar pattern.
%\todoscs{Explain or add ref for $\ell$-track graphs. See intro!}

\begin{thmrep}\applabel{thm:ltrackGraphs}
For each $\ell\in\mathbb{N}_{\geq 1}$, $\cG_{L_\ell^{\text{trk}}}$ is the class of $\ell$-track graphs, where  $L_{\ell}^{\text{trk}} \coloneqq \{ w \in \{0,1\}^{2\ell\text{-uni}}\mid \exists k\in [\ell]: \deletion{k}{k}{2\ell}(w) \in \langle 0101, 0110 \rangle\}$ is $2\ell$-uniform.
\end{thmrep}

\begin{proof}
By construction, $L_\ell^{\text{trk}}$ is $2\ell$-uniform.
\begin{enumerate}
\item Let $G = (V,E)$ be an  $\ell$-track graph with  representation $f_G: V \to \intervals^\ell$. For $i \in [\ell]$, $f_G(v)(i)$\longversion{, the $i$-th component of $f_{G}(v)$,} describes the interval corresponding to $v$ on track~$i$. 
W.l.o.g., assume that 
%the intervals in the $\ell$-track representation are pairwise disjoint for all $i \in [\ell]$ 
all intervals on the $i$-th track (for all $i \in [\ell]$) have pairwise different endpoints and that $\max \{\rborder(f_G(v)(i))\mid v\in V\}<\min\{\lborder(f_G(v)(i+1))\mid v\in V\}$\longversion{, i.e., 
all the endpoints of intervals on track~$i$ are smaller than the smallest endpoint of an interval on track $i+1$} for all $i\in[\ell-1]$. Let $x_1 < \dots < x_{2\ell|V|}$ be the endpoints of the intervals in the $\ell$-representation~$f_G$\longversion{, i.e., $\{x_i\mid i\in [2 \ell |V|]\}=\{\lborder(f_G(v)(i)),\rborder(f_G(v)(i))\mid v\in V, i\in [\ell]\}$}. 
Define  $w\in V^{2\ell |V|}$ by $w[i] \coloneqq v$ if $x_i\in \{\lborder(f_G(v)(j)),\rborder(f_G(v)(j))\mid j\in [\ell]\}$ is an endpoint of an interval in $f_G(v)$ for all $i \in [2 \ell |V|]$. 
%Let $w \coloneqq \prod_{i \in [2 \ell |V|]}w_i$. 
Clearly, $G(L_{\ell}^{\text{trk}},w) = G$. 
\item Let $G=(V,E) \in \mathcal{G}_{L_\ell^{\text{trk}}}$ be represented by a word~$w$. As $L_\ell^{\text{trk}}$ is $2\ell$-uniform, all $v \in \overline{V_{2\ell}(w)}$ are isolated vertices. For $v \in V_{2\ell}(w)$ and $m \in [\ell]$, define $I_{m,v} \coloneqq [\position_v(2m-1,w'),\position_v(2m,w')]$.
%For $v \in V_{2\ell}(w)$ and $m \in [\ell]$, define $I_{m,v} \coloneqq [e_1,e_2]$, where $e_1$ is the $(2m-1)$-th and the $2m$-th occurrence of $v$ in~$w$, respectively. 
For each $m\in [\ell]$ and for all $v \in \overline{V_{2\ell}(w)}$, choose an interval $I_{m,v}$ for~$v$ on track~$m$ such that the endpoints are greater than the largest endpoint of any interval assigned on track~$m$ for vertices from $V_{2\ell}(w)$, and such that all the intervals added to track~$m$ are pairwise disjoint. 
Let $f(v) \coloneqq \{I_{m,v} \mid m \in [\ell]\}$. Similarly to the proof of \myref{thm:lIntervalGraphs}, one sees that $f$ is an $\ell$-track representation of~$G$.\qed
\end{enumerate}
\end{proof}

As $\ell$-track graphs are a subclass of $\ell$-interval graphs, we get the following corollary that again solves  questions explicitly stated as unknown for $\ell=2,3$.
%in \url{https://graphclasses.org/classes/speed.html}.

\begin{corollary}\label{cor:l-track}
The class of  $\ell$-track graphs is factorial for each fixed $\ell\in\mathbb{N}_{\geq 1}$.
\end{corollary}

%A graph is a \emph{balanced 2-interval graph} \iffl it has an intersection model whose objects consist of two equal length intervals on a real line. Such graphs are 2-interval graphs, and hence their speed is also factorial, while this was stated as ``unknown'' in \url{https://graphclasses.org/classes/speed.html}. 

\longversion{\subsection{Boxicity}}
\noindent
We now turn to graphs defined by an intersection model of $b$-dimensional boxes.

%The concept of boxicity has been introduced by Roberts~\cite{Rob69}.  Boxicity refers to another geometric intersection model: now, $b$-dimensional axes-parallel boxes are associated to vertices. In particular, $b=2$ is also known as rectangle intersection graphs, %For the special class of graphs of boxicity at most~$2$, 
%we refer to~\cite{QueWeg90}.

\begin{thmrep}\applabel{thm:bboxGraphs}
For each $b\in\mathbb{N}_{\geq 1}$, $\cG_{L_b^{\text{box}}}$ is the class of graphs of boxicity~$\leq b$, where $L_{b}^{\text{box}} \coloneqq \{w \in \{0,1\}^{2b\text{-uni}}\mid \forall k\in [b]: \deletion{k}{k}{2b}(w) \in \langle 0101, 0110 \rangle\}$ is $2b$-uniform. 
\end{thmrep}

\begin{proof} 
Clearly, $L_{b}^{\text{box}}$ is $2b$-uniform by construction.
\begin{enumerate}
\item Let $G = (V,E)$ be a graph of boxicity $b$, i.e., there is a set of boxes $\{B_v\}_{v \in V}$ with $B_v = [\ell_1(v),r_1(v)] \times \dots \times [\ell_b(v),r_b(v)]$ for real numbers $\ell_d(v), r_d(v)$ for each $d \in [b]$ and $v \in V$. These boxes correspond to the vertices of $G$ in the sense that for all $u,v \in V$, $\{u,v\} \in E \Leftrightarrow B_v \cap B_u \neq \emptyset$. W.l.o.g., assume that all the $\ell_d(v)$ and $r_d(v)$ are distinct. Let $x_1 < \dots <x_{2b |V|}$ be the numbers $\{ \ell_d(v), r_d(v) \mid d \in [b], v \in V\}$ in linear order. Define $w\in V^{2b|V|}$ by setting, for all $i \in [2b|V|]$, $w[i] \coloneqq v$ if there is a $d \in [b]$ such that $x_i \in \{ \ell_d(v), r_d(v)\}$. %Let $w \coloneqq \prod_{i \in [2b|v|]} w_i$. 
Note that for each $d \in [b]$ and $v \in V$, $\ell_d(v)=\position_v(2d-1,w)$ and $r_d(v)=\position_v(2d,w)$.
Therefore, for all $u,v\in V$, 
\begin{align*}
h_{u,v}(w) \in L_b^{\text{box}} 
& \Leftrightarrow \forall d\in [b]: \deletion{d}{d}{2b}(w) \in \langle 0101, 0110 \rangle \\
& \Leftrightarrow \forall d\in [b]: [\ell_d(u),r_d(u)] \cap [\ell_d(v),r_d(v)] \neq \emptyset \\
& \Leftrightarrow B_u \cap B_v \neq \emptyset \\
& \Leftrightarrow \{u,v\} \in E.
\end{align*}
Hence, $G(L_b^{\text{box}},w) = G$.
\item Let $G =(V,E)\in \mathcal{G}_{L_{b}^{\text{box}}}$, i.e. there exists a word $w \in V^*$ such that $G(L_{b}^{\text{box}},w) = G$. 
%Let $V_{2b}\coloneqq \{ v \in V \mid |w|_v = 2b\}$. A
By $2b$-uniformity, all $v \in \overline{V_{2b}(w)}$ are isolated vertices. 
For $v \in V_{2b}(w)$ and $d \in [b]$, define 
$\ell_d(v)\coloneqq\position_v(2d-1,w)$ and $r_d(v)\coloneqq\position_v(2d,w)$.
%$I_{d}(v) \coloneqq [\ell_d(v),r_d(v)]$, where $\ell_d(v)$ is the $(2d-1)$-th and $r_d(v)$ the $2d$-th occurrence of $v$ in~$w$ and 
Let $B_v \coloneqq [\ell_1(v),r_1(v)] \times \dots \times [\ell_b(v),r_b(v)]$. For each $v \in \overline{V_{2b}(w)}$, we choose a box $B_v$ that is disjoint from all other boxes. Clearly, $f$ is a $b$-dimensional box representation of~$G$. Proof details are similar to \myref{thm:lIntervalGraphs}.\qed 
\end{enumerate}
\end{proof}

\begin{corollary}\label{cor:bboxGraphs}
The class of  graphs of boxicity~$\leq b$ is factorial for each %fixed
$b\in\mathbb{N}_{\geq 1}$.
\end{corollary}

\noindent
This result can also (independently) be concluded from \cite[Lemma 1]{ErdWes85}.

\longversion{\begin{remark}
Boxicity is only one example of a notion of dimensionality in the sense of Cozzens and Roberts, see \cite{CozRob89}. With similar constructions, we can also describe (for instance) graphs of bounded overlap dimension with a suitable language, based on \myref{exa:word-representability}, leading to $L_{b}^{\text{ovlp}} = \{w \in \{0,1\}^{2b\text{-uni}}\mid \forall k\in [b]: \deletion{k}{k}{2b}(w) \in \langle 0101\rangle\}$. As threshold graphs are interval graphs, the threshold dimension (also discussed, for instance, in~\cite{AdiBhoCha2010}) upper-bounds the boxicity of a graph. Hence, graphs of threshold dimension at most~$t$ have factorial speed.\footnote{This (classical) notion of threshold graph and threshold dimension should not be confused with the same notions introduced  in \cite{MolMurOel2020}; the context is metric dimension.} 
\end{remark}}

\subsection{Trapezoid Graphs and Variants}

In this section, we consider further variants of geometric intersection graphs\longversion{ that are related to trapezoid graphs, which is the class of graphs that we start with}.
\longversion{\subsubsection{Trapezoid and Point-Interval Graphs}}

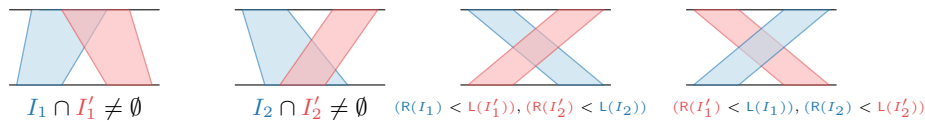
\begin{figure}[tb]
    \centering
    %\begin{subfigure}[t]
        %\centering
        \begin{tikzpicture}
            \draw[-] (0,0) -- (2,0); \draw[-] (0,-1) -- (2,-1);
            \filldraw[fill=LB,line width=0.5,draw=DB,opacity=0.6] (0.4,0) -- (1.3,0) -- (0.7,-1) -- (0.1,-1) -- (0.3,0);
            \filldraw[fill=LR,line width=0.5,draw=DR,opacity=0.6] (0.7,0) -- (1.6,0) -- (1.9,-1) -- (1.3,-1) -- (0.7,0);
            \node[] at (1,-1.3) {\small $\textcolor{DB}{I_1}\cap \textcolor{DR}{I_1'}\neq\emptyset$};

            \draw[-] (3,0) -- (5,0); \draw[-] (3,-1) -- (5,-1);
            \filldraw[fill=LB,line width=0.5,draw=DB,opacity=0.6] (3.1,0) -- (3.7,0) -- (4.5,-1) -- (3.4,-1) -- (3.1,0);
            \filldraw[fill=LR,line width=0.5,draw=DR,opacity=0.6] (4.3,0) -- (4.9,0) -- (4.2,-1) -- (3.6,-1) -- (4.3,0);
            \node[] at (4,-1.3) {\small $\textcolor{DB}{I_2}\cap \textcolor{DR}{I_2'}\neq\emptyset$};
            
            \draw[-] (6,0) -- (8,0); \draw[-] (6,-1) -- (8,-1);
            \filldraw[fill=LB,line width=0.5,draw=DB,opacity=0.6] (6.1,0) -- (6.7,0) -- (7.9,-1) -- (7.3,-1) -- (6.1,0);
            \filldraw[fill=LR,line width=0.5,draw=DR,opacity=0.6] (7.3,0) -- (7.9,0) -- (6.7,-1) -- (6.1,-1) -- (7.3,0);
            \node[] at (6.8,-1.3) {\tiny $\textcolor{DB}{(\rborder(I_1)}<\textcolor{DR}{\lborder(I_1'))}, \textcolor{DR}{(\rborder(I_2')}<\textcolor{DB}{\lborder(I_2))}$};

            \draw[-] (9,0) -- (11,0); \draw[-] (9,-1) -- (11,-1);
            \filldraw[fill=LR,line width=0.5,draw=DR,opacity=0.6] (9.1,0) -- (9.7,0) -- (10.9,-1) -- (10.3,-1) -- (9.1,0);
            \filldraw[fill=LB,line width=0.5,draw=DB,opacity=0.6] (10.3,0) -- (10.9,0) -- (9.7,-1) -- (9.1,-1) -- (10.3,0);
            \node[] at (10.5,-1.3) {\tiny $\textcolor{DR}{(\rborder(I_1')}<\textcolor{DB}{\lborder(I_1))}, \textcolor{DB}{(\rborder(I_2)}<\textcolor{DR}{\lborder(I_2'))}$};

        \end{tikzpicture}
    %\end{subfigure}
    \caption{Different patterns in the intersection model of two trapezoids. In  $L^{\text{trap}}$, the first two situations correspond to interval overlap and interval containment, expressed by projecting into $\langle 0101,0110\rangle$, while the last two correspond to $\langle 00111100\rangle$.}
    \label{fig:inter-trapezoid}
\end{figure}

\begin{thmrep}\applabel{thm:trapezoid}
%\shortversion{$(*)$}
For the 4-uniform language $L^{\text{trap}}$ defined as $$\{ w \in \{0,1\}^{4\text{-uni}} \mid \deletion{1}{1}{4}(w) \in \langle 0101, 0110 \rangle \vee \deletion{2}{2}{4}(w) \in \langle 0101, 0110 \rangle\} \cup \langle 00111100 \rangle\,,$$ $\cG_{L^{\text{trap}}}$ is the class of trapezoid graphs\shortversion{, which therefore has factorial speed}.
\end{thmrep}
\noindent
Different situations that $L^{\text{trap}}$ has to consider are shown in \myref{fig:inter-trapezoid}.

 \begin{proof}
Let $\cT$ be the set of trapezoids between two parallel lines. We can identify $\cT$ with $\intervals^2$.
%For each $T_1 \in \mathcal{T}$, let $U(T_1)$ be the interval on the upper line that describes the upper side of $T_1$ and let $L(T_1)$ be the interval on the lower line that describes the lower side of $T_1$. Notice that for any two trapezoids $T_1,T_2 \in \mathcal{T}$, $T_1 \cap T_2 \neq \emptyset$ \iffl $U(T_1) \cap U(T_2) \neq \emptyset$ or $L(T_1) \cap L(T_2) \neq \emptyset$ or $(\max U(T_1) < \min U(T_2) \wedge \max L(T_2) < \min L(T_1))$ or $(\max U(T_2) < \min U(T_1) \wedge \max L(T_1) < \min L(T_2))$.
Two trapezoids $\tau=(I_1,I_2)$ and $\tau'=(I_1',I_2')$ intersect \iffl $I_1\cap I_1'\neq\emptyset$ or $I_2\cap I_2'\neq \emptyset$ or $(\rborder(I_1)<\lborder(I_1'))\land (\rborder(I_2')<\lborder(I_2))$ or $(\rborder(I_1')<\lborder(I_1))\land (\rborder(I_2)<\lborder(I_2'))$. These different situations are illustrated in \myref{fig:inter-trapezoid}.
For each $d \in \mathbb{R}$, let $t_d: \cT \to \cT$ be such that, for any trapezoid $\tau=(I_1,I_2)$, % with $L(T_1) = [\ell_1, \ell_2]$, 
$t_d(\tau)=(I_1,[\lborder(I_2)+d,\rborder(I_2)+d])$.
%is the trapezoid described by $U(T_1)$ and $[\ell_1+d,\ell_2+d]$. 
Notice $(*)$ that for all $\tau_1,\tau_2 \in \cT$ and all $d \in \mathbb{R}$, $\tau_1$ and $\tau_2$ intersect \iffl 
$t_d(\tau_1)$ and $t_d(\tau_2)$ intersect.
%$T_1 \cap T_2 \neq \emptyset$ \iffl $t_d(T_1) \cap t_d(T_2) \neq \emptyset$.

\begin{enumerate}
\item Let $G$ be a trapezoid graph. There exists a set $T \subseteq \mathcal{T}$ such that $(T, \{\{\tau, \tau'\} \mid \tau, \tau' \text{ intersect}\,\}) \cong G$ and all the endpoints of the intervals describing the trapezoids in $T$ are distinct. Let $d\coloneqq \max_{(I_1,I_2)\in T}\rborder(I_1)+1$. Let $G_d$ be a trapezoid graph described by $T_d=\{t_d(\tau)\mid \tau \in T\}$. By $(*)$, $G\cong G_d$. 
%Choose a real number $d$ such that $d$ is greater than the maximum $m$ of the union of the intervals on the upper line of the trapezoids in $T$. Notice that $G$ is isomorphic to the graph $G_d = (T,E_d)$ described by $t_d(T)$ as well. 
Let $x_1 < \dots < x_{4|T|}$ be the endpoints of the intervals describing the trapezoids in $T_d$ in linear order; in other words, $x_i$ is the $x$-coordinate of a corner of a trapezoid from~$T_d$. Define $w\in V^{4|V|}$ by setting, for each $i \in [4|V|]$, $w[i]\coloneqq \tau$ if $x_i$ is an endpoint of an interval of the trapezoid~$\tau$. 
%Define $w \coloneqq \prod_{i \in [4|V|]} w_i$. 
For all $ \tau=(I_1,I_2),\tau'=(I_1',I_2')\in \cT$,
\begin{align*}
%\{T_1, T_2\} \in E_d 
\tau,\tau' \text{ intersect}
\Leftrightarrow \;  & (I_1\cap I_1'\neq \emptyset) \lor (I_2\cap I_2'\neq \emptyset)
%& U(T_1) \cap U(T_2) \neq \emptyset \\
%                & \vee L(T_1) \cap L(T_2) \neq \emptyset \\
\\&\lor ((\rborder(I_1)<\lborder(I_1'))\land (\rborder(I_2')<\lborder(I_2)))\\& \lor ((\rborder(I_1')<\lborder(I_1))\land (\rborder(I_2)<\lborder(I_2')))\\
%                & \vee (\max U(T_1) < \min U(T_2) \wedge \max L(T_2) < \min L(T_1)) \\
%                & \vee (\max U(T_2) < \min U(T_1) \wedge \max L(T_1) < \min L(T_2)) \\
\Leftrightarrow \;  & \deletion{1}{1}{4}(h_{\tau,\tau'}(w)) \in \langle 0101,0110\rangle  \lor \deletion{2}{2}{4}(h_{\tau,\tau'}(w)) \in \langle 0101,0110\rangle \\
                & \vee h_{\tau,\tau'}(w) = 00111100  \vee h_{\tau,\tau'}(w) = 11000011 \\
\Leftrightarrow \; & h_{\tau,\tau'}(w) \in L^{\text{trap}} 
\end{align*}
The second equivalence holds as all endpoints of $I_1$ and $I_1'$ are smaller than all endpoints of  $I_2$ and $I_2'$. Hence, $G \cong G_d \cong G(L^{\text{trap}},w)$.

\item Let $G = (V,E) \in \mathcal{G}_{L^{\text{trap}}}$, i.e., there exists a word $w \in V^*$ such that $G = G(L^{\text{trap}},w)$. %Let $V_4 = \{ v \in V \mid |w|_v = 4 \}$. 
For each $v \in V_4(w)$, let $v_1 < v_2 < v_3 < v_4$ be the indices of~$v$ in~$w$. Let $\tau_v=([v_1,v_2],[v_3,v_4])$.
%be the trapezoid with the interval $[u_1,u_2]$ on the upper line as the upper side and the interval $[\ell_1,\ell_2]$ on the lower line as the lower side of the trapezoid. 
For all $u \in \overline{V_4(w)}$, define trapezoids $\tau_u$ such that they do not intersect with any other trapezoid.
%are pairwise disjoint and disjoint to all the $\tau_v$ for $v \in V_4(w)$. 
We show analogously to the other inclusion that $\{\tau, \tau'\} \in E $ \iffl $h_{\tau,\tau'}(w) \in L^{\text{trap}} $. Hence, $G$ is the graph with the intersection model $\{\tau_v \mid v \in V\}\subset\cT$.\qed
\end{enumerate}
\end{proof}
\longversion{\begin{corollary}\label{cor:trapezoid}
The class of trapezoid graphs has factorial speed.
\end{corollary}}\shortversion{\noindent}
The point-interval graph classes $\textsc{PI}$ and $\textsc{PI}^*$  are well-known subclasses of trapezoid graphs.
\begin{toappendix}
More precisely, let $\cT$ be the set of triangles between two parallel lines $L$ and $M$ such that each $T \in \mathcal{T}$ has one corner $a(T)$ on $L$ and two corners $b(T)$ and $c(T)$ with $b(T) < c(T)$ on $M$.
The intersection model of these triangles defines the class \hypertarget{def:PI}$\textsc{PI}$. They are also called \emph{point-interval graphs} (which explains their abbreviations) because vertices are specified by one point on~$L$ and an interval on~$M$. If we allow, in addition, triangles with one point on $M$ and an interval on~$L$, then we arrive at the (larger) graph class $\textsc{PI}^*$. Clearly, both forms of triangles can be viewed as extreme cases of  trapezoids. This observation immediately entails a speed result for these graph classes.
\end{toappendix}
%Therefore, freely using the terminology from graphclasses.org, we can state:

\begin{corollary}\label{cor:PI}
The graph classes $\textsc{PI}$ and $\textsc{PI}^*$ have factorial speed.
\end{corollary}

We also have representations for these triangle (point-interval) intersection graph classes, which\longversion{ would yield the previous result as well and also} shows the versatility of our formalism. We show this for \textsc{PI}.

\begin{thmrep}\applabel{thm:PI-representation}
%Let $\delta_:\{0,1\}^{3\text{-uni}}\to \{0,1\}^{2\text{-uni}}$ be the function that deletes the first occurrence of~0 and the first occurrence of~1.
%%all but the 2nd and the 3rd occurrence of 0 and 1, respectively. 
%For the 3-uniform language $L^{\textsc{PI}} = \{ w \in 0^3 \shuffle 1^3 \mid \delta(w) \in \langle 0101, 0110 \rangle\} \cup \langle 011100 \rangle$, $\cG_{L^{\textsc{PI}}}$ is the class $\textsc{PI}$.
%\todohf{Easier to write this explicitly?}
%$L^{\textsc{PI}} =\langle 010101,100101,001101,010110,100110,001110, 011100 \rangle$.
%\\
%OR:\\
$L^{\textsc{PI}} \coloneqq\langle 001101,001110,010101,010110,011001,011010,011100 \rangle$ is a 3-uniform language such that $\cG_{L^{\textsc{PI}}}$ is the class $\textsc{PI}$.
\end{thmrep}

\begin{proof}
\begin{enumerate}
\item Let $\cT$ be the set of triangles between two parallel lines $L$ and $M$ such that each $T \in \mathcal{T}$ has one corner $a(T)$ on $L$ and two corners $b(T)$ and $c(T)$ with $b(T) < c(T)$ on $M$. We denote $T$ as $\Delta(a(T),b(T),c(T))$, and in this way, we can identify $\cT$ with $\mathbb{R}^3$. For each $d \in \mathbb{R}$, we define $t_d: \mathcal{T} \to \mathcal{T}$ with $t_d(T) = (a(T), b(T) + d, c(T) + d)$. Note that for each $T_1, T_2 \in \mathcal{T}$, $T_1 \cap T_2 \neq \emptyset$ \iffl $t_d(T_1) \cap t_d(T_2) \neq \emptyset$. 
Let $G = (V,E)$ be a \textsc{PI} graph described by $\cU \subset \cT$. W.l.o.g., assume that all the $x$-coordinates of the corners of the triangles are distinct. Let $m \coloneqq \text{max}_{T \in \mathcal{U}} a(T)+1$. 
Note that $\mathcal{U}_m \coloneqq \{t_m(T) \mid T \in \mathcal{U}\}$ describes $G$ as well. 
Let $x_1 < x_2 < \dots < x_{3|\mathcal{U}|}$ be the $x$-coordinates of the  corners of the triangles in $\mathcal{U}_m$. Define
$w\in V^{3|\cU|}$ by setting, for $i\in [3|\mathcal{U}|]$, 
%Let $w \coloneqq \Pi_{i \in [3|\mathcal{U}|]} w_i$ where 
$w_i \coloneqq T$ if $x_i$ is a corner of~$T$. 
For $T_1, T_2 \in \mathcal{U}_m$, $T_1 \cap T_2 \neq \emptyset$ \iffl $[b(T_1),c(T_1)] \cap [b(T_2),c(T_2)] \neq \emptyset$ or $a(T_1) < a(T_2) < b(T_2) < c(T_2) < b(T_1) < c(T_1)$ or $a(T_2) < a(T_1) < b(T_1) < c(T_1) < b(T_2) < c(T_2)$. 
Clearly, $T_1 \cap T_2 \neq \emptyset$ \iffl $h_{T_1,T_2}(w) \in L$. Hence, $G = G(L,w)$.
\item Let $G \in \mathcal{G}_L$, i.e. there exists a word $w \in V(G)^*$ such that $G(L^{\textsc{PI}},w) = G$. %Let $V_{3} \coloneqq \{ v \in V \mid |w|_v = 3 \}$. 
For each $v \in \overline{V_3}$, define $a_v,b_v,c_v > |w|$ such that the triangles $T_v  \coloneqq \Delta(a_v,b_v,c_v)$ are pairwise disjoint. 
For $v \in V_3$, let $T_v \coloneqq \Delta(\position_v(1,w), \position_v(2,w), \position_v(3,w))$. Clearly for $u,v \in V_2$, $T_v \cap \Delta_u \neq \emptyset$ iff $h_{u,v}(w) \in L^{\textsc{PI}}$ because all triangles $T_v$ have the corner points $a_v < b_v < c_v$, with all $a_v$ drawn on one line~$L$ and $b_v$, $c_v$ drawn on a parallel line~$M$. 
\qed
\end{enumerate}
\end{proof}

\begin{toappendix}
An alternative geometric reasoning is based on the idea that triangles can be viewed as a very special form of trapezoids, where the endpoints of one interval are identical (or at least very close to another).
Consider $w\in \{0,1\}^{\text{3-uni}}$. Apply to $w$ the operation $\iota:\{0,1\}^{\text{3-uni}}\to \{0,1\}^{\text{4-uni}}$ that replaces the first occurrence of~0 by $00$ and the first occurrence of~1 by $11$. Then, $w\in L^{\textsc{PI}}$ \iffl $\iota(w)\in L^{\text{trap}}$. Now, it is easy to check that 
%$\iota(011100)=00111100$, $\iota(010101)=00110101$, $\iota(100101)=11000101$, $\iota(0011$
\begin{align*}\iota(\{001101,001110,010101,010110,011001,011010,011100\}=\\\{00011101,00011110,00110101,00111001,00110110,00111010,00111100\}\subset L^{\text{trap}}\,.
\end{align*}
\end{toappendix}

\longversion{
\begin{theorem}\label{thm:PIstar-representation}
%Let $d$ be the map that deletes everything but the last three occurrences of $0$ and the last three occurrences of $1$.
$L^{\textsc{PI}^*} \coloneqq L^{\textsc{PI}} \cup \{ w \in \{ 0,1\}^{4-\text{uni}} \mid (d_0(d_1(w)))^R \in L^{\textsc{PI}}\} \cup  \langle \{w \in \{ 0,1\}^* \mid |w|_0 = 3 \wedge |w|_1 = 4 \wedge d_1(w) \notin (0 \cdot (00 \shuffle 11) \cdot 1) \cup 11 (0 \shuffle 1) \cdot 00)\} \rangle$ 
is a finite language such that $\cG_{L^{\textsc{PI}^*}}$ is the class $\textsc{PI}^*$.
\end{theorem}

Notice that we employ here the deletion operator $d_\ta(w)$ that deletes the first occurrence of~$\ta$ in~$w$ in our application, as we make sure that $|w|_\ta>0$.

\begin{proof}
Let $G = (V,E)$ be a $\textsc{PI}^*$ graph. Let $\Delta = (\Delta_v)_{v \in V}$ be an intersection model of $G$ such that for each $v \in V$, $\Delta_v$ is a triangle with one or two corners on a line $L$ and the other corner(s) on a parallel line $M$. W.l.o.g., assume that for each triangle $\Delta$ the corner(s) on $L$ are to the left of the corner(s) on $M$ and the positions of all corners are distinct. Let $p_1 < \dots < p_{3|V|}$ be the positions of the corners of the triangles from $\Delta$. Let $w \in V^{\text{3-uni}}$ such that $w[i] = v$ if $p_i$ is a corner of $\Delta_i$ for all $i \in [3|V|]$. Let $u \in V^*$ be such that for each $v \in V$, $|u|_v = 0$ if $\Delta_v$ is a triangle with one corner on $L$ and $|u|_v = 1$ if $\Delta_v$ is a triangle with two corners on $L$. We want to show that $G(L^{\textsc{PI}^*}, uw) = G$. Let $x,y \in V$. 
\begin{itemize}
    \item Case 1: $\Delta_x$ and $\Delta_y$ each have one corner on $L$. Hence, $h_{x,y}(uw) = h_{x,y}(w) \in \{0,1\}^{\text{3-uni}}$. We have to show that $\{x,y\} \in E$ \iffl $h_{x,y}(uw) \in L^{\textsc{PI}}$. This follows from \autoref{thm:PI-representation}.
    \item Case 2: $\Delta_x$ and $\Delta_y$ each have two corners on $L$. Hence, $h_{x,y}(uw) \in \{0,1\}^{\text{1-uni}} \cdot h_{x,y}(w) \subseteq \{0,1\}^{\text{4-uni}}$. We show that $\{x,y\} \in E$ \iffl $h_{x,y}(w)^R \in L^{\textsc{PI}}$. This follows from \autoref{thm:PI-representation}.
    \item Case 3: $\Delta_x$ has one corner on $L$ and $\Delta_y$ has two corners on $L$. Hence, $h_{x,y}(uw) = 1 \cdot h_{x,y}(w) \in 1 \cdot \{0,1\}^{\text{3-uni}}$. Clearly, $\{x,y\} \notin E$ \iffl $\Delta_x$ and $\Delta_y$ are disjoint. 
    \begin{itemize}
        \item Case 3.1: The leftmost corner of $\Delta_x$ and $\Delta_y$ belongs to $\Delta_x$. $\Delta_x$ and $\Delta_y$ are disjoint \iffl the corner of $\Delta_x$ on $L$ is to the left of all other corners and the corner of $\Delta_y$ on $M$ is to the right of all other corners \iffl $d_1(h_{x,y}(w)) \in (0 \cdot (00 \shuffle 11) \cdot 1)$. 
        \item Case 3.2: The leftmost corner of $\Delta_x$ and $\Delta_y$ belongs to $\Delta_y$. $\Delta_x$ and $\Delta_y$ are disjoint \iffl the corner of $\Delta_x$ on $L$ is to the right of the second corner of $\Delta_y$ on $L$ and the corner of $\Delta_y$ on $M$ is to the left of the first corner of $\Delta_x$ on $M$ \iffl $d_1(h_{x,y}(w)) \in (11 \cdot (0 \shuffle 1) \cdot 00)$.
    \end{itemize}
        Therefore, $\Delta_x$ and $\Delta_y$ are disjoint \iffl $h_{x,y}(w) \notin L^{\textsc{PI}^*}$.
    \item Case 4: $\Delta_x$ has two corners on $L$ and $\Delta_y$ has one corner on~$L$. This follows analogously to Case 3. 
\end{itemize}

For the converse direction, let $G = G(L^{\textsc{PI}},w)$ for some word $w$. Let $V_{3,4} := \{ \ta \in \alphabet(w) \mid 3 \leq |w|_\ta \leq 4\}$. Let $v$ be the word that we get by projecting $w$ to $V_{3,4}$. We will show that $G(L,v)$ is a $\textsc{PI}^*$ graph. This is sufficient because the class of $\textsc{PI}^*$ graphs is closed under adding isolated vertices. 

For each $\ta \in \alphabet(v)$ such that $|v|_\ta = 3$, let $\Delta_\ta$ be the triangle with one corner on $L$ at position $\position_\ta(1,w)$ and two corners on $M$ at the positions $\position_\ta(2,w)$ and $\position_\ta(3,w)$. 
For each $\ta \in \alphabet(v)$ such that $|v|_\ta = 4$, let $\Delta_\ta$ be the triangle with two corners on $L$ at the positions $\position_\ta(2,w)$ and  $\position_\ta(3,w)$ and one corner on $M$ at position $\position_\ta(4,w)$.
We want to show that $G(L^{\textsc{PI}},v)$ is the $\textsc{PI}^*$ graph with the intersection model $(\Delta_\ta)_{\ta \in \alphabet(v)}$. We can show that $h_{x,y}(v) \in L^{\textsc{PI}}$ \iffl $\Delta_x \cap \Delta_y \neq \emptyset$ with a case distinction analogous to the other direction. \qed
\end{proof}
}

%\todoscs{I believe for $\textsc{PI}^*$ we might need the prefix normal form. HF: This would not matter, as it will go into the appendix anyways.}

\noindent
We are now going to generalize \myref{thm:trapezoid} to a multi-dimensional scenario. %\todo{There seem to be problems with Cref and thmrep; hence, I introduced the macro myref}
%Let $L_1, \dots, L_d$ be $d$ parallel lines in the plane. A \emph{$d$-trapezoid} is the polygon obtained by choosing an interval $I_i$ on every line $L_i$ and connecting the left, resp. right endpoint of $I_i$ with the left, resp. right endpoint of $I_{i+1}$. A graph is a $d$-trapezoid graph if it has an intersection model consisting of $d$-trapezoids between $d$ parallel lines.

%\begin{definition}
%For every $\ell\in\mathbb{N}_{\geq 1}$ and  
%$k,m,k',m' \in [\ell]$ such that $k \neq k'$ and $m \neq m'$, let 
%$\delta^{2\ell}_{k,m,k',m'}:\{0,1\}^{2\ell\text{-uni}}\to \{0,1\}^{4\text{-uni}}$
%be the function that deletes any but the $\nth{(2k-1)}$, $\nth{(2k)}$, $\nth{(2k'-1)}$ and $\nth{(2k')}$ occurrence of~$0$ and any but the $\nth{(2m-1)}$, $\nth{(2m)}$, $\nth{(2m'-1)}$ and $\nth{(2m')}$ occurrence of $1$ from~$w$. 
%\end{definition}
%\todoscs{Do we want to simplify this or do we want to reuse this somewhere else?}

\begin{theorem} \label{thm:d-trapezoid-representation}
For every $d\in\mathbb{N}_{\geq 1}$, the finite language $L^{\text{trap}}_d = \{ w \in \{0,1\}^{2d\text{-uni}} \mid \exists 1 \leq i \leq d: \deletion{i}{i}{2d}(w) \in \langle 0101,0110 \rangle \vee \exists 1 \leq i \leq d-1: \deletion{i}{i}{2d\to 4}
%\delta^{2d}_{i,i,i+1,i+1} 
\in \langle 00111100 \rangle \}$ 
is $2d$-\linebreak[4]uniform\shortversion{;}\longversion{, and} $\mathcal{G}_{L^{\text{trap}}_d}$ is the class of $d$-trapezoid graphs\shortversion{ and hence has factorial speed}.
\end{theorem}

\begin{proof}[sketch]
Let $T,T'$ be two $d$-trapezoids and for $i \in [d]$, let $I_i$ be the interval of $T$ and $I_i'$ be the interval of $T'$ on the $\nth{i}$ line. Note that $T \cap T'\neq \emptyset$ \iffl there exists $1 \leq i \leq d -1$ such that the trapezoid described by $I_i$ and $I_{i+1}$ and the trapezoid described by $I'_i$ and $I'_{i+1}$ intersect. With this idea, the theorem can be proven analogously to \myref{thm:trapezoid} which handles the case $d=2$.\qed
\end{proof}

\longversion{\begin{corollary}
\label{cor:d-trapezoid-speed}    
For every $d\in\mathbb{N}_{\geq 1}$, the class of $d$-trapezoid graphs  has factorial speed.
\end{corollary}}

\longversion{\begin{remark}
Golumbic, Rotem and Urrutia~\cite{GolRotUrr83} have generalized permutation graphs by associating to every vertex a sequence of $d+1$ points, each on one of $d+1$ subsequent parallel lines $L_0,L_1,\dots,L_d$; these $d+1$ points define a piecewise linear curve. Two vertices are adjacent \iffl their curves intersect.
They showed that a graph~$G$ can be represented by an intersection model of curves with $d$ linear pieces \iffl its complement is a comparability graph  of a poset of dimension~$d$. Clearly there is a finite language representation for such a class. Also, as this intersection model can be seen as a ``very slim'' $(d+1)$-dimensional trapezoids, we get at most factorial speed for these classes immediately from \myref{cor:d-trapezoid-speed}.  
\end{remark}}

\begin{toappendix}
\shortversion{\paragraph{Further related graph classes.}}\longversion{\subsubsection{Generalizing Circle  Trapezoids}}
A \emph{circle trapezoid} is the region that lies between two non-crossing chords of a circle. So if $ab$ and $cd$ are non-crossing chords, with their endpoints in the order $abcd$
 on the circle, the boundary of the circle trapezoid $abcd$
 consists of the chord $ab$, the arc $bc$, the chord $cd$ and the arc $da$. Alternatively, a circle trapezoid is the convex hull of two disjoint arcs on the circle.
A graph is a \emph{circle-trapezoid graph} \cite{Lin2006} if it is the intersection graph of circle trapezoids on a common circle. These graphs naturally generalize trapezoid graphs. For general remarks about relations between linear and circular models, we refer to~\cite{KraKloMul97}.
Circle trapezoids are  the same as circle-2-gon graphs.
\end{toappendix}
%\todo{definitions of n-gon  are added at the beginning}
\shortversion{\noindent Circle-trapezoid ($=$ circle-2-gon) graphs generalize trapezoid graphs.}\longversion{They have been further generalized to circle-$k$-gon graphs. Another graph class that can be generalized to the circle-$k$-gon graphs are the $k$-gon-circle graphs.\longversion{ They are also interesting, because 3-gon-circle graphs clearly encompass $\textsc{PI}^*$. Also, it was shown in \cite{EnrKit2019} that the only interesting graph class that is both polygon-circle and word-representable (in the classical sense) are the circle graphs, which can be also viewed as 2-gon-circle graphs and as circle-1-gon graphs.}

\begin{thmrep} \applabel{thm:k-gon-circle-representation}
For every $k \geq 2$, let
$$
L^{\text{gon-c}}_k \coloneqq  \{w\in \{ 0,1\}^* \mid 2 \leq |w|_0, |w|_1  \leq k\} \setminus \langle 0^* 1^* 0^* \rangle.
$$
$\mathcal{G}_{L^{\text{gon-c}}_k}$ is the class of $k$-gon-circle graphs. 
\end{thmrep}

\begin{proof}
Let $c$ be a fixed point on the line of a circle and $G$ be a $k$-gon-circle graph described by convex polygons $p_1, \dots, p_n$ whose corners are on the circle line. W.l.o.g., assume that the corners of the polygons are distinct. Let $0 < c_1 < c_2 < \dots < c_\ell \leq 2 \pi$ be the positions of the corners of $p_1, \dots, p_n$ relative to $c$ with $c$ at position $0$. We call $w = p(c_1) \cdots p(c_\ell)$ with $p(c_i) = p_j$ if $c_i$ is a corner of $p_j$ for $i \in [\ell]$ and $j \in [n]$ the \emph{corner string} of the polygons. Note that $w \in \{\{ 0,1\}^* \mid 2 \leq |w|_0, |w|_1  \leq k\}$. For $p_i,p_j$ with $i,j \in [n]$, $p_i$ and $p_j$ intersect iff $w$ projected to $p_i$ and $p_j$ contains the subsequence $p_i p_j p_i p_j$ or $p_j p_i p_j p_i$. This is the case iff $h_{p_i,p_j}(w) \notin \langle 0^* 1^* 0^* \rangle $. Therefore, $G = G(L^{\text{gon-c}}_k,w)$.

Let $\mathcal{G}_{L^{\text{gon-c}}_k}$, i.e., there exists a word $w \in V(G)^*$ such that $G(L^{\text{gon-c}}_k,w) = G$. For each $\ta \in V(G)^*$ such that $|w|_\ta \leq k$, let $p_\ta$ be the polygon with corners at the positions $\{ \frac{2i\pi}{|w|+1} \mid w[i] = \ta\}$ on the line of a circle. For each $\ta,\tb \in V(G)$ such that  $|w|_\ta \leq k$, $p_\ta$ and $p_\tb$ intersect \iffl neither $\ta\tb\ta\tb$ nor $\tb\ta\tb\ta$ is a subsequence of $w$ \iffl $h_{\ta,\tb}(w) \notin L^{\text{gon-c}}_k$. Since the $k$-gon-circle graphs are closed under adding isolated vertices, $G$ is a $k$-gon-circle graph. \qed
\end{proof}
}

\begin{thmrep} \applabel{thm:circle-k-gon-representation}
For every $k\in\mathbb{N}_{\geq 1}$, let 
\begin{align*}
L^{\text{c-gon}}_k \coloneqq \{ w \in \{ 0,1\}^* \mid{} & 2 \leq |w|_0, |w|_1 \leq 2k+1 \wedge \longversion{{}\\
&} (\text{$|w|_0$ and $|w|_1$ odd} \\
& \vee (\text{$|w|_0$ odd and $|w|_1$ even} \Rightarrow w \notin (00)^+(11)^*0^+)\\
& \vee (\text{$|w|_0$ even and $|w|_1$ odd} \Rightarrow w \notin (11)^+(00)^*1^+) \\
& \vee (\text{$|w|_0$ and $|w|_1$ even} \Rightarrow w \notin \langle  (00)^+ 1^* 0^*\rangle)
\}
\end{align*}
%where $\text{del}_{\text{odd}}$ deletes the first occurrence of $0$ if the number of $0$s is odd and the first occurrence of $1$ if the number of $1$s is odd. 
$\mathcal{G}_{L^{\text{c-gon}}_k}$ is the class of circle-$k$-gon graphs plus isolated vertices. \shortversion{Hence, the class of circle-$k$-gon graphs has factorial speed.}
\end{thmrep}

\begin{proof}
%[sketch]
We fix a point $c$ on the circle. For $1 \leq k' \leq k$, describe the circle-$k'$-gons~$\Delta$ by the $2k'$ endpoints of the chords of $\Delta$ on the circle. We assume that no two chords share an endpoint. (If they do, we replace the endpoint by two points directly next to each other.) If one side of $\Delta$ is an arc that contains $c$, we denote $\Delta$ by $2k'+1$ letters and the last $2k'$ letters describe the endpoints of the chords on the circle. The position of the first letter is ignored. The language $L^{\text{c-gon}}_k$ describes all possibilities for a circle-$k_1$-gon $\Delta_1$ and a circle-$k_2$-gon $\Delta_2$, $1\le k_1,k_2\le k$, such that all the endpoints of chords of $\Delta_2$ are in between two arcs of $\Delta_1$. For an example, see \autoref{fig:circle-n-gon}. More precisely, this figure indicates the general situation that $\Delta_1$ and $\Delta_2$ do not intersect. 

%First consider the case that $c$ is not contained in any arc of $\Delta_1$ or $\Delta_2$. 
%This means that, for the letter (vertex) $\ta_i$ that is represented by $\Delta_i$ (for $i=1,2$), the projection of a word~$w$ that describes a graph~$G$ in this scenario contains $2k_i$ occurrences of $\ta_i$. In the projection $h_{\ta_1,\ta_2}(w)$, assuming that $\position_{\ta_1}(1,w)<\position_{\ta_2}(1,w)$, which means that we first find an endpoint of $\Delta_1$ when moving from $c$ on the cycle in clockwise direction, we find $h_{\ta_1,\ta_2}(w)\in (00)^+1^*0^*$ \iffl $\Delta_1$ and $\Delta_2$ do not intersect. Symmetrically, if $\position_{\ta_2}(1,w)<\position_{\ta_1}(1,w)$, we find $h_{\ta_1,\ta_2}(w)\in (11)^+0^*1^*$ \iffl $\Delta_1$ and $\Delta_2$ do not intersect. 
%Next, assume that $c$ is contained in an arc of $\Delta_1$ but not of $\Delta_2$. This means that $w$ contains $2k_1+1$ occurrences of $\ta_1$ and $2k_2$ occurrences of $\Delta_2$. 

Let $G = (V,E)$ be a circle-$k$-gon graph with additional isolated vertices. Let $V'$ be the set of non-isolated vertices in $G$. Let $c$ be a point on the circle that is not an endpoint of any interval. Since $G[V']$ is a circle-$k$-gon graph, it is the intersection graph of the circle-$k$-gons associated with the $k$-gon-interval-filaments (union of arcs) with respect to $c$, i.e., for each vertex $v \in V'$ there is an interval-filament $a(v) \subseteq [0,2\pi]$ and there exists a set of intervals $I(v) = \{ I_{1,v}, \dots, I_{k_v,v}\}$ such that $a(v) = \bigcup_{j \in [k_v]} I_{j,v}$ with $k_v\le k+1$. (The intervals $I_{j,v}$ describe arcs on the circle in clockwise direction from~$c$ to~$c$.) Two vertices $u,v \in V'$ are adjacent iff the associated circle-$k$-gons intersect. W.l.o.g., assume that all endpoints of the intervals are distinct except for $0$ and $2\pi$, if $\lborder(I_{1,v}) = 0$ and $\rborder(I_{k_v,v}) = 2\pi$. Let $k' \coloneqq \sum_{v \in V'} k_v - |\{ v \in V' \mid \lborder(I_{1,v}) = 0 \wedge \rborder(I_{k_v,v})=2\pi\}| $. Let $e_1 < \dots < e_{2k'}$ be the endpoints of the intervals in ascending order except for $0$ and $2\pi$. Define the word $w'$ by $w'[j] \coloneqq v$ if $e_j$ is an endpoint of an interval of $a(v)$ for $v \in V'$. 
Further, let $w_\mathrm{l}$ be a word that contains every $v\in V'$ with $\lborder(I_{1,v}) = 0$ and $\rborder(I_{k_v,v})=2\pi$ exactly once, and let $w_\mathrm{r}$ be a word that contains every isolated vertex exactly once.
We define $w=w_\mathrm{l}\cdot w'\cdot w_\mathrm{r}$ (i.e., the concatenation of the three words). We will show that $G(L^{\text{c-gon}}_k,w) = G$. Clearly, all $v \in V \setminus V'$ are isolated vertices in $G(L^{\text{c-gon}}_k,w)$ as well. Clearly, for all $v \in V'$, $|w|_v$ is odd iff $c \in a(v)$. 
Let $u,v \in V'$. 
\begin{enumerate}
\item Case 1: $|w|_v$ and $|w|_u$ are even. 
\begin{align*}
h_{u,v}(w) \notin L^{\text{c-gon}}_k 
& \Leftrightarrow h_{u,v}(w) \in \langle (00)^+ 1^* 0^* \rangle \\
& \Leftrightarrow \{u,v\} \notin E
\end{align*}
For the last equivalence, consider the left side of \autoref{fig:circle-n-gon}.
\item Case 2: $|w|_v$ is odd and $|w|_u$ is even.
\begin{align*}
h_{u,v}(w) \notin L^{\text{c-gon}}_k 
& \Leftrightarrow h_{u,v}(w) \in (00)^+ (11)^* 0^+\\
& \Leftrightarrow h_{u,v}(w') \in 0(00)^* (11)^* 0^+\\
& \Leftrightarrow \{u,v\} \notin E
\end{align*}
For the last equivalence, consider the right side of \autoref{fig:circle-n-gon}.
\item Case 3: $|w|_v$ is even and $|w|_u$ is odd. This case is analogous to the previous case.
\item Case 4: $|w|_v$ and $|w|_u$ are odd. In this case, $h_{u,v}(w) \in L^{\text{c-gon}}_k$. By construction, $c \in a(v) \cap a(u)$ and therefore $\{u,v\} \in E$. 
\end{enumerate}
In each case, $h_{u,v}(w) \in L^{\text{c-gon}}_k \Leftrightarrow \{u,v\} \in E$. 

For the other direction, let $G \in \mathcal{G}_{L^{\text{c-gon}}_k}$, i.e. there exists a word $w$ such that $G(L^{\text{c-gon}}_k,w)$. Let $V' \coloneqq \{ v \in V \mid 2 \leq |w|_v \leq 2k+1\}$. Each $v \in G(V) \setminus V'$ is isolated. Let $v \in V'$. Let $v_1 < \dots < v_{|w|_v} $ be the indices of letters $v$ in $w$. If $|w|_v$ is even, let $I_{j,v} \coloneqq [2j-1, 2j]$ for each $j \in [\frac{|w|_v}{2}]$. If $|w|_v$ is odd, let $I_{1,v} \coloneqq [0,v_2]$, $I_{j,v} \coloneqq [2j-1, 2j]$ for each $1 < j < \frac{|w|_v-1}{2}$ and $I_{\frac{|w|_v-1}{2},v} \coloneqq [v_{|w|_v}, 2\pi]$. In both cases, let $a(v) \coloneqq \bigcup_{j \in [\lfloor\frac{|w|_v}{2}\rfloor]} I_{j,v}$.
Let $G = (V',E)$ be the circle-$k$-gon graph that is the intersection graph of the circle-$k$-gon associated with the $k$-gon-interval-filaments $a(v)$ for $v \in V'$. Analogously to the other direction, we prove that $h_{u,v}(w) \in L^{\text{c-gon}}_k$ iff $\{ u,v\} \in E$ for all $u,v \in V'$.  \qed
\end{proof}

\begin{figure}
\begin{minipage}[b]{.31\textwidth} %
\begin{center}
\scalebox{.8}{\begin{tikzpicture}
	\newcommand{\rad}{2cm}

	\begin{pgfonlayer}{background}
	   \node[draw,circle,minimum size=2*\rad] (C) at (0,0) {};
	\end{pgfonlayer}
	\node[circle, draw, fill=white, inner sep=1.5pt, label={[shift={(190:2mm)}, yshift=-.5mm, anchor=center, color=DB]{0}}] (w_0) at (C.190) {};
	\node[circle, draw, fill=white, inner sep=1.5pt, label={[shift={(220:2.5mm)}, anchor=center, color=DB]{0}}] (w_1) at (C.220) {};
	\node[circle, draw, fill=white, inner sep=1.5pt, label={[shift={(240:3mm)}, anchor=center, color=DR]{1}}]  (w_2) at (C.240) {};
	\node[circle, draw, fill=white, inner sep=1.5pt, label={[shift={(260:3.5mm)}, anchor=center, color=DR]{1}}]  (w_3) at (C.260) {};
	\node[circle, draw, fill=white, inner sep=1.5pt, label={[shift={(330:2.5mm)}, anchor=center, color=DR]{1}}]  (w_4) at (C.330) {};
	\node[circle, draw, fill=white, inner sep=1.5pt, label={[shift={(350:2mm)}, anchor=center, color=DR]{1}}]  (w_5) at (C.350) {};
	\node[circle, draw, fill=white, inner sep=1.5pt, label={[shift={(10:2mm)},  yshift=-.5mm, anchor=center, color=DB]{0}}] (w_6) at (C.10) {};
	\node[circle, draw, fill=white, inner sep=1.5pt, label={[shift={(60:4mm)},  yshift=-2mm, anchor=center, color=DB]{0}}] (w_7) at (C.60) {};
	\node[circle, draw, fill=white, inner sep=1.5pt, label={[shift={(120:2mm)}, anchor=center, color=DB]{0}}] (w_8) at (C.120) {};
	\node[circle, draw, fill=white, inner sep=1.5pt, label={[shift={(150:2mm)}, anchor=center, color=DB]{0}}] (w_9) at (C.150) {};
	
	\begin{pgfonlayer}{background}
    	\draw[line width=0.5mm,color=DB] (w_0.center) arc (190:220:\rad) -- (w_6.center) arc (10:60:\rad) -- (w_8.center) arc (120:150:\rad) -- (w_0.center);
    	\draw[line width=0.5mm,color=DR] (w_2.center) arc (240:260:\rad) -- (w_4.center) arc (330:350:\rad) -- (w_2.center);
	\end{pgfonlayer}
	
	\node[label={[shift={(85:4mm)}, yshift=-2mm, anchor=center]{c}}] (c) at (C.85) {};
	\draw (C.85) -- (85:2.2cm);
	
	\node[yshift=-8mm, anchor=center] (w) at (C.270) {(a) \textbf{{\color{DB}00}{\color{DR}1111}{\color{DB}0000}}};
\end{tikzpicture}}    
\end{center}
\end{minipage} %
\begin{minipage}[b]{.31\textwidth} %
\begin{center}
\scalebox{.8}{\begin{tikzpicture}
	\newcommand{\rad}{2cm}

	\begin{pgfonlayer}{background}
	   \node[draw,circle,minimum size=2*\rad] (C) at (0,0) {};
	\end{pgfonlayer}
	\node[circle, draw, fill=white, inner sep=1.5pt, label={[shift={(190:2mm)}, yshift=-.5mm, anchor=center, color=DB]{0}}] (w_0) at (C.190) {};
	\node[circle, draw, fill=white, inner sep=1.5pt, label={[shift={(220:2.5mm)}, anchor=center, color=DB]{0}}] (w_1) at (C.220) {};
	\node[circle, draw, fill=white, inner sep=1.5pt, label={[shift={(240:3mm)}, anchor=center, color=DR]{1}}]  (w_2) at (C.240) {};
	\node[circle, draw, fill=white, inner sep=1.5pt, label={[shift={(260:3.5mm)}, anchor=center, color=DR]{1}}]  (w_3) at (C.260) {};
	\node[circle, draw, fill=white, inner sep=1.5pt, label={[shift={(330:2.5mm)}, anchor=center, color=DR]{1}}]  (w_4) at (C.330) {};
	\node[circle, draw, fill=white, inner sep=1.5pt, label={[shift={(350:2mm)}, anchor=center, color=DR]{1}}]  (w_5) at (C.350) {};
	\node[circle, draw, fill=white, inner sep=1.5pt, label={[shift={(10:2mm)},  yshift=-.5mm, anchor=center, color=DB]{0}}] (w_6) at (C.10) {};
	\node[circle, draw, fill=white, inner sep=1.5pt, label={[shift={(60:4mm)},  yshift=-2mm, anchor=center, color=DB]{0}}] (w_7) at (C.60) {};
	\node[circle, draw, fill=white, inner sep=1.5pt, label={[shift={(120:2mm)}, anchor=center, color=DB]{0}}] (w_8) at (C.120) {};
	\node[circle, draw, fill=white, inner sep=1.5pt, label={[shift={(150:2mm)}, anchor=center, color=DB]{0}}] (w_9) at (C.150) {};
	
	\begin{pgfonlayer}{background}
    	\draw[line width=0.5mm,color=DB] (w_0.center) arc (190:220:\rad) -- (w_6.center) arc (10:60:\rad) -- (w_8.center) arc (120:150:\rad) -- (w_0.center);
    	\draw[line width=0.5mm,color=DR] (w_2.center) arc (240:260:\rad) -- (w_4.center) arc (330:350:\rad) -- (w_2.center);
	\end{pgfonlayer}
	
	\node[label={[shift={(30:4mm)}, yshift=-2mm, anchor=center]{c}}] (c) at (C.30) {};
	\draw (C.30) -- (30:2.2cm);
	
	\node[yshift=-8mm, anchor=center] (w) at (C.270) {(b) \textbf{{\color{DB}(0)0}{\color{DR}1111}{\color{DB}00000}}};
\end{tikzpicture}}
\end{center}\end{minipage} %
\begin{minipage}[b]{.36\textwidth}
The figures on the left illustrate two main cases leading to $L^{\text{c-gon}}_n$ in \myref{thm:circle-k-gon-representation}. If $c$ is chosen so that no arc contains it, as in (a), then the vertex representation contains an even number of letters; otherwise, as in (b), it contains an odd number, but the first one is ignored concerning the positions of the letters. Other cases are symmetric or even trivial.
\end{minipage}
\caption{A circle-3-gon and a circle-2-gon\longversion{ (circle trapezoid) described}\shortversion{ given} by a word, for different positions of\longversion{ the point}~$c$.}
\label{fig:circle-n-gon}
\end{figure}

\longversion{\begin{remark}
The reader might have asked what the interrelation between between $k$-gon-circle and  circle-$k$-gon graphs is. Clearly, every $k$-gon-circle graph is a circle-$k$-gon graph. But also conversely, any circle-polygon graph is a polygon-circle graph, which can be seen by replacing each arc by a sequence of polygon edges that approximate the arc in the sense that these edges will make the resulting polygon intersect with the other geometric figures whenever the arc did so.
\end{remark}}

\longversion{\begin{corollary}
\label{cor:circle-k-gon-speed}    
For every $k\in\mathbb{N}_{\geq 1}$, the class of circle-$k$-gon graphs  has factorial speed.
\end{corollary}}
%\todo{The following paragraph (def.) is not really used so far.}
\shortversion{\noindent}
%As circular-arc graphs are circle-1-gons, \todo{Gavril writes circle graphs, not circular-arc graphs! Somebody should check our construction for $n=1$.}
%a subclass of circle-trapezoid graphs, i.e. the circle-2-gon graphs, 
%we can conclude\longversion{ (again fixing a class marked as ``unknown'' in graphclasses.org)}:

%\begin{corollary}\label{cor:circular-arc}
%The class of circular-arc graphs is factorial.
%\end{corollary}

%There is another way to deduce \myref{cor:circular-arc}, namely, i
It is shown in \cite{GamHabVia2007,GamVia2007} that circular-arc graphs are contained in the class of balanced 2-interval graphs, so that factorial speed can be concluded with \myref{cor:interval}. 
\longversion{\begin{corollary}\label{cor:circular-arc}
The class of circular-arc graphs is factorial.
\end{corollary}}
\begin{comment}
\begin{thmrep}\applabel{thm:circular-arc-as-2-interval}
Every circular-arc graph is a 2-interval graph.
\end{thmrep}
\begin{proof} 
Let $G=(V,E)$ be given by a collection of arcs of the unit circle~$C$. 
If one cuts the circle arbitrarily, say, at the point $(1,0)$, then one can consider it as an interval $I_C$ of length $2\pi$. Those arcs that do not contain the point $(1,0)$ as an inner point can be viewed as ``single subintervals'' of $I_C$. The arcs that contain $(1,0)$ appear as two intervals, one of the form $[0,a]$ and another one of the form $[b,2\pi]$.  If we now split the ``single subinterval'' $[p,r]$ by choosing any $q\in (p,r)$ and take $[p,q]$ and $[q,r]$ instead, one understands that we can associate to each $v\in V$ a collection of two intervals such that the originally associated arcs intersect \iffl 
the newly associated pairs of intervals intersect.\qed 
\end{proof}    
\end{comment}
%\todohf{Many classes starting with circle .. or circular .. in our ``speed unknown'' file should work with our approach. More comments there ... \\ Silas: Is there any class from the file that we want to add to this work before we publish it on arXiv?}
%\todohf{Def. of  $k$-circular track graphs? I hope the following description suffices.}

\longversion{We will provide an explicit finite-language representation of circular-arc graphs in the next subsection. 
\subsubsection{Generalizations and Restrictions of Circular-arc Graphs}}

\emph{$\ell$-circular track graphs} generalize circular-arc graphs similarly as $\ell$-track graphs generalize interval graphs. Hence, combining the arguments that show that $\ell$-track graphs are $\ell$-interval graphs with the rather straightforward arguments showing that circular-arc graphs are 2-interval graphs,
\longversion{also see \cite[Propriété~2]{GamHabVia2007},}%\myref{thm:circular-arc-as-2-interval}, 
we can conclude:

\begin{theorem}\label{thm:k-circular-track-as-2k-interval}
For $k\geq 1$, every  $k$-circular track graph is a $2k$-interval graph.
\end{theorem}

\begin{corollary}
    \label{cor:k-circular-track}
For each $k \geq 1$, the class of $k$-circular track graphs is factorial.
\end{corollary}

\longversion{Of course, one can also give a language representation of $k$-circular track graphs, but as this is now more like an exercise, we do not make it explicit here.}

\begin{remark}\label{rem:interval-containment}
There are interesting subclasses of circular-arc graphs, i.e., 1-circular track graphs, with hitherto unknown speed like co-interval containment bigraphs.\longversion{ This also settles the question of whether the complement class of interval containment bigraphs are factorial.} Interval containment bigraphs are equivalent to  complements of circular-arc graphs of clique covering number 2~\cite{HelHua2004}, bipartite co-circular-arc graphs~\cite{ShrTayUen2010}, two-directional orthogonal-ray graphs~\cite{ShrTayUen2009,ShrTayUen2010}, and permutation bigraphs~\cite{SahBSW2014}\longversion{, also compare \cite{HelHLM2020} for a more general setting}.
\end{remark}
%\todohf{Possibly, ``our way'' to deal with complements could help for    co-trapezoid graphs. Again, it is open if their speed is factorial.}

\longversion{In \cite{FraGonOch2015}, the authors introduced \emph{$t$-circular interval graphs}. In this intersection model, to each vertex, a collection of $t$ arcs of one circle is associated. Clearly, $1$-circular interval graphs are circular-arc graphs. %Let $m_1: \{0,1\}^* \to \{0,1\}^*$ be the map that moves the first $1$ in a word to the leftmost position, e.g., $m_1(00{\color{blue}1}01) = {\color{blue}1}0001$. 
\begin{theorem}
For each $t \geq 1$, the language 
\begin{align*}
L^{\text{c-int}}_t \coloneqq \langle \{w \in \{0,1\}^* \mid{} & 2t \leq |w|_0, |w|_1 \leq 2t+1 \wedge  {}\\& (|w|_0 = |w|_1 = 2t+1 \vee {}\\
& (|w|_0 = |w|_1 = 2t \wedge \exists i,j \in [t]: \deletion{i}{j}{2t}(w) \in \langle 0101,0110 \rangle ) \vee {} \\
& (|w|_0 = 2t \wedge |w|_1 = 2t+1 \wedge \exists i,j \in [t+1]: \\
& \deletion{i}{j}{2(t+1)}(1d_1(w)100) \in \langle 0101,0110 \rangle)\,)\} \rangle
\end{align*}
represents the $t$-circular interval graphs plus isolated vertices. 
\end{theorem}
\begin{proof}[sketch]
We fix a point $c$ on the circle line. If a letter $\ta \in \alphabet(w)$ occurs $2t$ times, we interpret the positions of $\ta$ in $w$ as the endpoints of the arcs listed around the circle in clockwise order starting after the point $c$ (with no arc containing $c$). If a letter $\ta \in \alphabet(w)$ occurs $2t+1$ times, we ignore the first occurrence of $\ta$ and interpret the other occurrences in such a way that one of the arcs contains the point~$c$. 
Note that if two sets of arcs both contain an arc that contains~$c$, they intersect. This explains the condition $|w|_0 = |w|_1 = 2t+1$ in $L^{\text{c-int}}_t$. If two sets of arcs both do not contain an arc that contains~$c$, they intersect \iffl there is one arc in each of them such that those two arcs intersect. If one set of arcs contains an arc that contains $c$ and the other set of arcs does not contain an arc that contains $c$, we interpret the arc that contains $c$ as two arcs separated by $c$ and add an arc that intersects none of the arcs of the other set of arcs. This is done by using the word $1d_1(w)100$ instead of~$w$ as $d_1$ deletes the first occurrence of~1. Therefore, the language $L^{\text{c-int}}_t$ describes exactly the $t$-circular interval graphs plus isolated vertices.\qed
\end{proof}

For $t=1$, we get an explicit finite language representation of circular-arc graphs. More precisely, 
$$L^{\text{c-int}}_1 = \{ 0,1 \}^\text{3-uni} \cup \langle 0101,0110 \rangle \cup \langle (00 \shuffle 111) \setminus \{(1)1001\} \rangle.$$ Note that $(1)1001$ describes the only case of non-intersecting intervals for $0$ and~$1$ such that $c$ is contained in the interval for $1$ and not contained in the interval for~$0$. 
}

\longversion{\subsection{Interval Enumerable Graphs}}

\shortversion{\noindent}Yet another generalization of interval graphs are interval enumerable graphs\longversion{, introduced in}~\cite{BanHuaYeo2000}. 
\begin{toappendix}
Obviously, a linear ordering $<$ on a finite set~$M$ can be specified by writing a word of length~$|M|$ over the alphabet~$M$ that lists each letter once. Also, after fixing such an ordering, for $s,t\in M$ and $s\leq t$, we can define the \emph{interval} $[s,t]\coloneqq\{q\in M\mid s\leq q\leq t\}$.  
A graph $G$ is  \emph{interval enumerable} if there exists a linear ordering  $v_1v_2\cdots v_n$ of $V (G)$ such that, for each vertex $v_i$, there exist numbers $1 \leq \ell_i, r_i \leq  i-1$ such that $N(v_i) \cap [v_1, v_{i-1}] = [v_{\ell_i}, v_{r_i}]$, where $\ell_i > r_i$ is used to indicate that $v_i$ has no neighbor preceding it in the ordering. That is, we require that for each $v_i$ the neighbors with labels less than $i$ form an interval.  
For example, the 4-cycle $\ta-\tb-\tc-\td-\ta$ is interval enumerable by the ordering $\ta\tc\tb\td$. It can be described by the word $\ta\ta\tb\td\ta\tc\tc\tc\td\tb\tb\td$ in the following sense.
\end{toappendix}
\begin{thmrep}
\applabel{thm:int-enum}
For the 3-uniform language $L^{\text{int-en}}\coloneqq\langle (00\shuffle 1)\cdot \{011\}\rangle$, $\cG_{L^{\text{int-en}}}$ is the class of interval enumerable graphs\shortversion{, which hence is factorial}.
\end{thmrep}

%\longversion{\noindent In the construction, we see that $\ta\ta\tb\td\ta\tc\tc\tc\td\tb\tb\td$ describes a $C_4$.}

\begin{proof}
 Clearly, the language  $L^{\text{int-en}}$ is 3-uniform.
\begin{enumerate}
    \item Let $G=(V,E)$ be interval enumerable, i.e., there is a linear  ordering  $v_1v_2\cdots v_n$ of $V$ such that, for each vertex $v_i$, there are numbers $1 \leq \ell_i, r_i \leq  i-1$ such that $N(v_i) \cap [v_1, v_{i-1}] = [v_{\ell_i}, v_{r_i}]$. We construct a 3-uniform word $w_n$ from this ordering inductively as follows. We always claim that $G[\{v_1,\dots,v_j\}]\cong G(L^{\text{int-en}},w_j)$. Let $w_1=v_1v_1v_1$. Assume $w_j$ describes $G[\{v_1,\dots,v_j\}]$ for $j<k$ and construct $w_k$ as follows. If  $\ell_k > r_k$, then $w_k=w_{k-1}v_kv_kv_k$. Otherwise, let $w_{k-1}=u_1u_2u_3$, where $|u_1|_{v_{\ell_k}}=2$, $u_2[1]=v_{\ell_k}$, $u_2[|u_2|]=v_{r_k}$, $|u_1u_2|_{v_{r_k}}=3$. Define $w_k=u_1v_ku_2v_ku_3v_k$. By construction, for any $i<k$, $h_{v_i,v_k}(w_k)$ will end in $011$ \iffl $\ell_k\leq i\leq r_k$.  Thus inductively, $G\cong G(L^{\text{int-en}},w_n)$.
    \item Let $w$ be a 3-uniform word. Let $V=\alphabet(w)$. Then, we claim that $G=G(L^{\text{int-en}},w)$ is interval enumerable. Consider $V$ ordered by the sequence of last (i.e., third) occurrences of letters in~$w$.  Consider a situation where $x<y<z<v$  in this ordering, with $\{x,z\}\subseteq N(v)$. This means that $h_{x,v}(w)$ and $h_{z,v}(w)$ both end with $011$, as $x<z<v$. Hence, in $w$, we find $w=u_1xu_2zu_3$ with $|u_1|_x=2$, $|u_1u_2|_z=2$, $|u_2|_v=0$, $|u_3|_v=2$. As the third occurrence of $z$ must be within $u_2$, $h_{y,v}(w)$ also ends with $011$. Hence, $G$ is  interval enumerable.\qed 
\end{enumerate}  
\end{proof}

As convex-round graphs are a subclass of interval enumerable graphs and as concave-round graphs are complements of convex-round graphs\longversion{, see}~\cite{BanHuaYeo2000},\longversion{ and because biconvex graphs (known to be factorial) are a subclass of convex-round graphs,} we can deduce\shortversion{ factorial speed also for these classes.}\longversion{:

\begin{corollary}\label{cor:int-enum}
The classes of  interval enumerable graphs, convex-round graphs and of concave-round graphs have factorial speed.
\end{corollary}
An alternative argument can be based on \cite[Theorem 3.1]{BanHuaYeo2000} where it is shown that concave-round graphs are circular-arc graphs, and then applying \myref{cor:k-circular-track}.} Conversely, proper circular-arc graphs are concave-round, which shows the lower bound for \shortversion{this class}\longversion{\myref{cor:int-enum}}.
\longversion{An interesting fact about convex-round graphs is that o}\shortversion{O}f the 11 minimal superclasses that are listed under graphclasses.org, all but one are known to have superfactorial speed; this only graph class with\longversion{ hitherto} unknown speed is \shortversion{$\cG_{L^{\text{int-en}}}$}\longversion{the class of interval enumerable graphs}.

% Convex-round graphs have T_2 as a forbidden induced subgraph, which is a 1-subdivision of a claw, so a tree.
% We get the forbidden induced subgraph from the complement of those of concave-round graphs.
% Reference: https://onlinelibrary.wiley.com/doi/pdf/10.1002/jgt.22486

\subsection{Thinness and Lettericity}

All previously considered graph classes can be viewed as generalizing the intersection model of interval graphs. \longversion{In this subsection, we are leaving the world of intersection graph models. W}\shortversion{Now, w}e will rather look at another characterization of interval graphs that we then generalize to $k$-thinness.
We \longversion{will }relate the $k$-thin graphs to two further graph classes: they contain the $k$-letter graphs and are contained in the graphs of boxicity~$k$. The latter result is implicitly claimed in Table~2 of \cite{BeiCKMMS2024} with a reference to an unpublished manuscript; graphclasses.org lists \cite{ManORC2007} as a reference, but this paper does not contain a proof. \longversion{Therefore, w}\shortversion{W}e provide an argument \longversion{here}\shortversion{in the appendix}.

\begin{toappendix}
\begin{figure}\centering
    \includegraphics{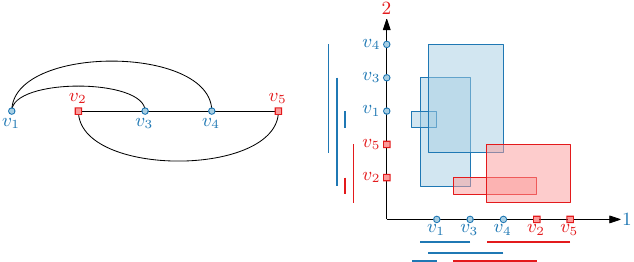}
    \caption{A boxicity-2 representation of a 2-thin graph. Partition $V^1$ is drawn in blue, partition $V^2$ is drawn in red.}
    \label{fig:thin-boxicity}
\end{figure}
 \end{toappendix}

\begin{thmrep}
\applabel{thm:thin-boxicity}
\longversion{For each $k\in\mathbb{N}_{\geq 1}$, the class of $k$-thin graphs is a subclass of the graphs of boxicity~$k$.}\shortversion{For $k\in\mathbb{N}_{\geq 1}$, $k$-thin graphs have boxicity at most~$k$.}
\end{thmrep}

%\todohf{This theorem might be known, as a similar relation is implicitly claimed in Table~2 of \cite{BeiCKMMS2024}; they refer for explanations to their own ArXiv version. There is a bigger table on page 28. BUT: exactly the entry important to us points to an unpublished manuscript that they only seem to know of via yet another paper that cites it. To Philipp: I assume you know at least some of the authors of this chaotic reference, maybe you can follow this up with our best Christmas wishes?}

\begin{proof}
Let $G=(V,E)$ be a $k$-thin graph with a partition $V=V^1\uplus V^2\uplus\cdots\uplus V^k$. We want to find a representation of the vertices by axis-parallel $k$-dimensional boxes in $\mathbb{R}^k$ that intersect \iffl an edge between the vertices exists. If we restrict the axis-parallel $k$-dimensional boxes to one dimension $i$, we obtain an interval representation of some graph $H^i=(V,E^i)$. We have that two boxes intersect \iffl the corresponding intervals overlap in each dimension, that is, we obtain a boxicity $k$ representation of the graph $(V,E^1\cap\ldots\cap E^k)$. See \myref{fig:thin-boxicity} for an example.

For every vertex $v_b$ and every partition $i$, let $\ell_b^i$ be the left-most neighbor of $v_b$ in $V^i$, that is, the vertex $v_a\in V^i$ with $\{v_a,v_b\}\in E$ such that $a$ is minimal. Note that $\ell_b^i$ might not exist for some combinations of $b$ and $i$.

We will now describe the interval representation in dimension $i$ and the graph $H^i$ that it represents. We order the right endpoints of the intervals as follows. We first place the right endpoints of the vertices in $V^i$, ordered by the linear ordering of the graph, and then the right endpoints of all other vertices ordered by the linear ordering of the graph (note that only the order for the vertices in $V^i$ matters, all other vertices could be placed in any arbitrary order). For every vertex $v_a \in V^i$, we place the left endpoint directly left of $\ell_a^i$, if it exists and $\ell_a^i<v_a$, or directly left of $v_a$, otherwise. For every other vertex $v_b \in V^j, j\neq i$, we place the left endpoint directly left of $\ell_b^i$, if it exists, and directly left of the leftmost vertex that is not in $V^i$, otherwise. Thus, the graph $H^i$ will contain the following edges:
\begin{enumerate}[label=(\roman*)]
    \item the edges $\{v_a,v_b\}$ with $v_a \in V^i$ and $a<b$, if $\ell_b^i=v_a$ or $\ell_b^i<v_a$;
    \item the edges $\{v_a,v_b\}$ with $v_a\in V^i$ and $v_b\notin V^i$, if $\ell_b^i=v_a$ or $\ell_b^i<v_a$;
    \item the edges $\{v_a,v_b\}$ with $v_a,v_b\notin V^i$.
\end{enumerate}

We will now prove that $G=(V,E^1\cap\ldots\cap E^k)$. First, let $\{v_a,v_b\}\in E$ with $v_a,v_b\in V^i$ and $a<b$. By definition, either $\ell_b^i=v_a$ or $\ell_b^i<v_a$, so by (i) we have $\{v_a,v_b\}\in E^i$ and by (iii) we have $\{v_a,v_b\}\in E^j$ for every $j\neq i$. Second, let $\{v_a,v_b\}\in E$ with $v_a\in V^i, v_b\in V^j, j\neq i$. Again, either $\ell_b^i=v_a$ or $\ell_b^i<v_a$, so by (ii) we have $\{v_a,v_b\}\in E^i$ and $\{v_a,v_b\}\in E^j$, and by (iii) we have $\{v_a,v_b\}\in E^h$ for every $h\neq i,j$. Finally, let $\{v_a,v_b\}\notin E$ with $v_a\in V^i$ and $a < b$. If $\ell_b^i<v_a$, then the $k$-thinness implies that $\{v_a,v_b\}\in E$; a contradiction. Hence, by (i) or (ii) $\{v_a,v_b\}\notin E^i$. \qed
\end{proof}

\begin{theorem}
For all $k\in\mathbb{N}_{\geq 1}$, any $k$-letter graph is a  $k$-thin graph.
\end{theorem}

\begin{proof}
Let $G = ([n],E)$ be a $k$-letter graph. Hence, there is a word $w \in [k]^*$ and a decoder $\cD \subseteq [k]^2$ 
%There is an alphabet $\Sigma$ with $|\Sigma| = k$, a word $w \in \Sigma^*$ and  a decoder $\mD \subseteq \Sigma^2$ 
such that $G \cong G(\cD, w)$. Let $P_\ta \coloneqq \{ j \in [n] \mid w[j] = \ta \}$ for each $\ta \in [k]$. Clearly, $(P_\ta)_{\ta\in [k]}$ is a partition of~$V$. Let $r,s,t \in [n]$ with $r < s < t$, $r,s \in P_\ta$ for a letter $\ta \in [k]$ and $\{r,t\} \in E$. The latter implies $w[r] w[t] \in \cD$. Hence, $w[s] w[t] \in \cD$, since $w[s] = \ta = w[r]$. Thus, $\{s,t \} \in E$ and $G$ is $k$-thin.\qed 
\end{proof}

\begin{corollary}\label{cor:k-thin}
For each $k\in\mathbb{N}_{\geq 1}$, the class of $k$-thin graphs and the class of $k$-letter graphs have at most factorial speed.
\end{corollary}

As\longversion{ the} $1$-thin graphs are exactly the interval graphs, $k$-thin graphs are\longversion{ indeed} factorial.
%\todo{Somebody should check the following; also, I did not detect anything with my claims in \cite{Loz2023}, but this paper is quite densely written and is the only one I know of that both discusses speed and lettericity, only in two different sections, as it appears to be the case.}
By a\longversion{ basic} combinatorial argument, we show that $k$-letter graphs have a smaller speed.

\begin{thmrep}
\applabel{thm:kletter-speed}
For\longversion{ each} $k\in\mathbb{N}_{\geq 2}$, \longversion{the class of }$k$-letter graphs ha\longversion{s}\shortversion{ve} exponential speed.
\end{thmrep}

\begin{proof}
Let  $k\in\mathbb{N}_{\geq 2}$ be fixed. Obviously, $|[k]^n|=k^n$ is the number of different words of length~$n$ on $k$ letters. For specifying a $k$-letter graph, apart from a word $w\in [k]^n$, one needs a decoder, a set of 2-letter words over the alphabet $[k]$, i.e., $\cD\subseteq {[k]^2}$.
Hence, in total, there cannot be more than $2^{k^2}\cdot k^n$ labeled $k$-letter graphs of order~$n$. As $k$ is fixed, this speed is at most exponential.

Recall that a graph is a \emph{complete split graph} if its vertex set can be partitioned in an independent set and a clique such every vertex in the independent set is adjacent to every vertex in the clique. Complete split graphs are known to have exponential speed (more precisely, $2^n$) which is readily seen, as we can assign any vertex either to the independent set or to the clique.\footnote{Here, it is important that we consider labeled speed.}
Now, consider $w\in[2]^n$ and $\cD=\{11,12,21\}$. Then, the positions (vertices) labeled~1 will form a clique and the vertices labeled~2 will form an independent set, while every vertex in the independent set is adjacent to every vertex in the clique. This way, each complete split graph is seen to be a $k$-letter graph for each $k>1$.
\qed 
\end{proof}

\begin{toappendix}
The case $k=1$ is excluded as the speed of 1-letter graphs is even constant. More precisely, it is two: the complete and the null graph of order~$n$ are the only such graphs, see \cite{Pet2002}. By \cite[Theorem~12]{SchZit94}, this is the largest hereditary class with constant speed.
\end{toappendix}
\shortversion{\noindent}As already interval graphs have factorial speed, we see:

\begin{corollary}\label{cor:k-letter-thin}
\longversion{For each $k\in\mathbb{N}_{\geq 1}$, t}\shortversion{T}he $k$-letter graphs form a proper subclass of the $k$-thin graphs.
\end{corollary}

\begin{toappendix}
We can also provide an explicit family of examples to prove the strictness of this inclusion. Clearly, every path is an interval graph, i.e., it is 1-thin. However, by \cite{Fer2020}, the paths $P_{3k-2}$ are $k$-letter graphs but not $(k-1)$-letter graphs.
\end{toappendix}

\begin{toappendix}
In \cite{FenFMRS2025}, the definition of $k$-letter graphs has been extended 
towards so-called generalized $k$-letter graphs. They allow for more general decoders. Nonetheless, even then the argument from \myref{thm:kletter-speed} remains valid, so that we conclude:

\begin{corollary}
For each fixed $k\in\mathbb{N}_{\geq 2}$, the class of generalized $k$-letter graphs has exponential speed.
\end{corollary}

The case $k=1$ is a bit different, because now the fact that the class of generalized $k$-letter graphs is not hereditary comes into play insofar as the speed of generalized $1$-letter graphs is linear, see \cite[Proposition~4]{FenFMRS2025}. 
\end{toappendix}

\subsection{Containment Graphs\shortversion{ and (Co-)comparability Graphs}}

A different way to define graphs (and hence graph classes) from geometric objects is behind the idea of containment graphs: Here, an edge is drawn between vertices $u$ and $v$ if either the object associated to~$u$ is contained in the object associated to~$v$ or vice versa. Starting again with intervals, we arrive at interval containment graphs\longversion{ in this way. In fact, they are}\shortversion{,} better known as permutation graphs. %We are considering several containment graph models in this section.
%For instance, a 
A graph is a \emph{containment graph of circular arcs} if it has a containment model consisting of arcs on a circle. This graph class is equivalent to the class of circular permutation graphs\longversion{~\cite{RotUrr82}, as shown in~\cite{NirMasNak88}}\shortversion{ \cite{NirMasNak88,RotUrr82}}.\longversion{ By \myref{rem:adding-isolates}, the following theorem shows that circular permutation graphs have factorial speed.}

\begin{thmrep}
\applabel{thm:containment-circular-arcs}
There is a finite language $L=\langle 0110,100110,010110,01100,\linebreak[4] 01011,00111 \rangle$ such that $\cG_L$ is the class of containment graphs of circular arcs, plus isolated vertices. \shortversion{Hence,  $\cG_L$ has factorial speed.}
\end{thmrep}

\begin{proof}[sketch]
Let $G=(V,E)$ be given as a containment graphs of circular arcs. Formally, this means that we associate to $v\in V$ a pair of numbers $(x_v,y_v)\in [0,2\pi]^2$ such that if $x_v<y_v$, this refers to the arc
starting at $(\cos(x_v),\sin(x_v))$ and spanning anti-clockwise until $(\cos(y_v),\sin(y_v))$ on the unit circle; if   $x_v>y_v$, we move clockwise from $(\cos(y_v),\sin(y_v))$ to $(\cos(x_v),\sin(x_v))$ on the unit circle to describe the circular arc. We can assume that all numbers in $P_G=\{x_v,y_v\mid v\in V\}$ are different, i.e., $|P_G|=2|V|$,  and that $0\notin P_G$. Then, the arcs described by $(x_v,y_v)$ with $x_v>y_v$ are precisely those arcs containing the point $(1,0)$. Let $u\in V^*$ be any word that contains exactly one letter $v$ for every $v\in V$ with  $x_v>y_v$. Initialize $w=u$. Then, we start from $(1,0)$ and move along the unit circle in anti-clockwise direction and we append $v$ to $w$ whenever we encounter an endpoint (i.e., $(\cos(x_v),\sin(x_v))$ or $(\cos(y_v),\sin(y_v))$) of the arc associated to~$v$ on our journey around the circle, until we are back at $(1,0)$. This describes the word~$w$. It can be shown in general that $G\cong G(L,w)$.
For instance, $w=\ta\tb\ta\tc\td\td\tb\tb\tc\ta$ describes a $P_4$; we might have found $x_\ta=\frac{16\pi}{9}$, $y_\ta=\frac{2\pi}{9}$, $x_\tb=\frac{12\pi}{9}$, $y_\tb=\frac{10\pi}{9}$, $x_\tc=\frac{4\pi}{9}$, $y_\tc=\frac{14\pi}{9}$, $x_\td=\frac{6\pi}{9}$, $y_\td=\frac{8\pi}{9}$, so that we have the edges $\{\ta,\tb\}$, $\{\tb,\td\}$ and $\{\tc,\td\}$ in the intersection model. Notice that $h_{\ta,\tb}(w)=101001=\widetilde{010110}$, $h_{\ta,\tc}(w)=00110\notin L$, $h_{\tb,\td}(w)=01100$, $h_{\tc,\td}(w)=0110$. 
As usual, additional isolated vertices can be modeled by letters that occur, say, four times. 

Conversely, consider a word $w\in V^*$. Any letter $v\in V$ with $|w|_v\notin\{2,3\}$ describes an isolated vertex. Assume now that $w$ contains only letters that occur twice or thrice. Consider the 2-uniform word $w'$ obtained from $w$ by deleting the first occurrence of any letter that occurs thrice.
Associate to $v\in V$ the pair of numbers $(x_v,y_v)$ with $x_v=\frac{2\position_v(1,w')\pi}{2|V|+1}$, $y_v=\frac{2\position_v(2,w')\pi}{2|V|+1}$ if $|w|_v=2$ and with   $x_v=\frac{2\position_v(2,w')\pi}{2|V|+1}$, $y_v=\frac{2\position_v(1,w')\pi}{2|V|+1}$ if $|w|_v=3$. This describes an arc containment model of $G(w,L)$.
\qed  
\end{proof}
\longversion{
\begin{corollary}
\label{cor:containment-circular-arcs}
The class of containment graphs of circular arcs ha\longversion{s}\shortversion{ve} factorial speed.
\end{corollary}}

\begin{remark}
According to \href{http://graphclasses.org/classes/gc_136.html}{graphclasses.org/classes/gc\_136.html}, 
the question of proper containment of containment graphs of circular arcs in several super-classes like comparability graphs seems to be open. As the speed of all these super-classes is  superfactorial, we have settled proper containment.
\end{remark}

\longversion{\subsection{(Co-)comparability Graphs of Fixed Dimension Posets}}
\label{subsec:orderings}

As containment between geometrical objects obviously  yields a partial ordering on the set of these objects, it is quite natural go go a step further from containment graphs to comparability graphs\longversion{, where the considered partial ordering no longer needs to be geometric. This section is hence devoted to the study of comparability graphs of partial orders of dimension~$d$}. 
\begin{toappendix}
A first idea would be to simply write the $d$ linear orders on~$V$ one after the other to obtain a word over~$V$ that describes a comparability graph of the given partial order of dimension~$d$.
This would lead us to the language $\langle (01)^d\rangle$. 
A special case (for $d=2$) would be the overlap graphs, also known as circle graphs. More generally speaking, we would arrive at $L_{\classical}\cap \{0,1\}^{d\text{-uni}}$. 
However, as we will see in the proof, we need a different set of languages.
\end{toappendix}
To this end, let us first inductively define \emph{well-parenthesized words} (wpw) by saying:
\begin{itemize}
    \item $()$ is a wpw with one pair of parentheses.
    \item If $w$ is a wpw with $d$  pairs  of parentheses, then $(w)$ is a wpw with $(d+1)$  pairs  of parentheses.
    \item If $w_1$ and $w_2$ are wpw with $d_1$ and $d_2$  pairs  of parentheses, respectively, then $w_1w_2$ is a wpw with $(d_1+d_2)$  pairs  of parentheses.
 \longversion{   \item No other words over the alphabet $\{(,)\}$ are wpw.}
\end{itemize}
Let $L_d^{\text{Dyck}}$ be the language of all wpw with exactly $d$  pairs of parentheses. The language of all wpw (with any number of parentheses) is also known as the language of restricted Dyck words, see \cite{Ber79}. 
Let  $L_d^{\text{cmp}}=\langle h_{(,)}(L_d^{\text{Dyck}})\rangle$. For example,
 $L_3^{\text{cmp}}=\langle 000111,001011,001101,010011,010101\rangle$\longversion{ as $((()))$, $(()())$, $(())()$, $()(())$ and $()()()$ are the wpw with 3 pairs of parentheses}. 
 To our knowledge, the speed of comparability graphs of fixed dimension $d\ge 3$ had not previously been settled.

\begin{thmrep}\applabel{thm:dimGraphs}
For each $d\in\mathbb{N}_{\geq 2}$, \longversion{the class of }comparability graphs of posets of dimension~$d$ can be represented by the\longversion{ $d$-uniform} language $L_d^{\text{cmp}}$. \shortversion{Hence, $\cG_{L_d^{\text{cmp}}}$ is factorial.}
\end{thmrep}

\begin{proof} 
Clearly, $L_d^{\text{cmp}}$ is $d$-uniform.
\begin{enumerate}
    \item 
If $G=(V,E)$ is a  comparability graph of a poset~$P\subseteq V\times V$ of dimension~$d$, this means that there are $d$ linear orderings $L_1,\dots,L_d\subseteq V\times V$ such that $\bigcap_{i\in [d]}L_i=P$. Let $n=|V|$ and fix a bijection $f:[n]\to V$. A linear ordering $L_i$ (for any $i\in [d]$) can be expressed as a 1-uniform word $w_i=w_{i,1}\cdots w_{i,n}$ such that, for all $j,k\in [n]$, $w_{i,j}$ precedes $w_{i,k}$ in $w_i$ \iffl $(f(j),f(k))\in L_i$. Hence, for all $j,k\in [n]$, $(f(j),f(k))\in P$ \iffl $w_{i,j}$ precedes $w_{i,k}$ in all $w_i$ (with $i\in [d]$). Now, let $w=w_1\cdots w_d$. Consider some vertices $u,v\in V$. Then $uv$ is an edge in $G(w,L_d^{\text{cmp}})$ iff for $j=f^{-1}(u)$ and $k=f^{-1}(v)$, either $w_{i,j}$ precedes $w_{i,k}$ in any $w_i$ (as this is equivalent to $h_{u,v}(w)=(01)^d$), or $w_{i,k}$ precedes $w_{i,j}$ in any $w_i$ (as this is equivalent to $h_{v,u}(w)=(01)^d$). Hence, $uv$ is an edge in $G(w,L_d^{\text{cmp}})$ \iffl either $(u,v)\in P$ or $(v,u)\in P$, i.e., if $u,v$ are comparable in~$P$.

\item
Conversely, let $w\in V^*$. Let $G=(L_d^{\text{cmp}},w)=(V,E)$. By $d$-uniformity of $L_d^{\text{cmp}}$, any letters that do not occur exactly~$d$ times in~$w$ correspond to isolated vertices in~$G$. Let $V_I$ be the collection of all these isolated vertices and define $V'=V\setminus V_I$ and $w'\in V^{\prime\,*}$ as the word obtained from $w$ by erasing all letters from $V_I$. Hence, $G'=(L_d^{\text{cmp}},w')=G[V']$ and by construction, $w'$ is $d$-uniform.  
With the help of the function $\position_\ta(i,w)$, we extract an ordering $L_i$ on $V'$ as follows: $(u,v)\in L_i$ if $\position_u(i,w)\leq \position_v(i,w)$.
Now, consider an edge $uv\in E$, or, equivalently, $h_{u,v}(w)\in L_d^{\text{cmp}}$.
Notice that  $\position_u(i,w)\leq \position_v(i,w)$ \iffl  $\position_0(i,h_{u,v}(w))\leq \position_1(i,h_{u,v}(w))$ for all $i\in [d]$, which in turn is equivalent to $h_{u,v}(w)\in L_d^{\text{cmp}}$. \qed 
%We now define the $d$-linearization $d\text{-lin}(w')$ as follows. For $d=1$, $d\text{-lin}(w')=w'$. For $d>1$, $d\text{lin}(w')=w_{\text{first}}(d-1)\text{lin}(w'')$, where $w_{\text{first}}$ lists all letters from $V'$ exactly once, in the order in which they appear first in~$w'$, and $w''$ is a $(d-1)$-uniform word obtained from $w'$ by deleting all first occurrences of any letter.
%Obviously, $h_{u,v}(w')\in L_d^{\text{cmp}}$ means $h_{u,v}(w')=h_{u,v}(d\text{-lin}(w'))\in L_d^{\text{cmp}}$. Assume that $h_{u,v}(w')\notin L_d^{\text{cmp}}$. Walking from left to right over  $h_{u,v}(w')$, there must be a first occurrence of $00$ or of $11$. Then, at the very same position, there is also an occurrence of $00$ or of $11$ in $h_{u,v}(d\text{-lin}(w'))$. Hence, $h_{u,v}(d\text{-lin}(w'))\notin L_d^{\text{cmp}}$. This is remarkable, because in general $h_{u,v}(w')\neq h_{u,v}(d\text{-lin}(w'))$ in that case. For instance, consider the $3$-uniform word $000111$ \todo{This does not work, sorry ...}
\end{enumerate}
\end{proof}

\longversion{\begin{corollary}\label{cor:ddimGraphs}
For each $d\in\mathbb{N}_{\geq 2}$, the class of comparability graphs of posets of dimension~$d$ has factorial speed.
\end{corollary}}

\longversion{\begin{remark}
From the considerations in \cite{KitSei2008,Sei2008}, one can deduce that each comparability graph of a poset of dimension~$d$ can be permutationally represented by a $d$-uniform word. Translated into our framework, this means that $\cG_{L_d^{\text{cmp}}}\subseteq \cG_{\langle (01)^d\rangle}$.
This can be seen as a refinement of the fact that comparability graphs are word-representable in the classical sense and was already made explicit in \cite[Lemma~4]{HalKitPya2011}. Reference~\cite{MozKri2025} additionally discusses some relations to boxicity.
\end{remark}}

\longversion{\begin{remark}\label{rem:planar}
Schnyder~\cite{Sch89} suggested to associate to a graph $G=(V,E)$ a strict order of height~2 on the set $X=V\cup E$ by setting $x<y$ \iffl $x\in V$, $y\in E$ and $x$ is incident to~$y$. He could prove that $G$ is planar \iffl the corresponding partial order $\leq$ on~$X$ has dimension at most~3. As the comparability graph of~$\leq$ is basically the incident graph of~$G$, with vertex set~$X$, and as planar graphs of order~$n$ have incident graphs of order at most $4n$, the fact (deduced from \myref{cor:ddimGraphs}) that there are factorially many comparability graphs of partial orders of dimension at most~3 implies that  planar graphs have at most factorial speed.  This is interesting as it proves that sometimes, our approach must necessarily employ inclusion arguments, because it is possible to show that there is no language representation of planar graphs, see~\cite{FenFFMS2024}.
\end{remark}}

\longversion{Although the definition of interval dimension looks more difficult at first glance when compared to that of a dimension of a poset, we will see that now, the simple idea discussed in the beginning of this section works out.}
The definition of comparability graphs of bounded  interval dimension can be easily re-interpreted as a multi-interval representation. \longversion{Namely, l}\shortversion{L}et $d\geq 1$.
$G=(V,E)$ is a comparability graph of posets of bounded interval dimension~$d$ \iffl we can associate to each vertex~$v$ a \longversion{collection}\shortversion{set} of $d$ intervals $I_{v,i}$, $i\in [d]$, with $I_{v,i}=[\ell_{v,i},r_{v,i}]$, such that $uv\in E$ \iffl either $r_{u,i}\leq \ell_{v,i}$ for all $i\in [d]$ or  $r_{v,i}\leq \ell_{u,i}$ for all $i\in [d]$. 
\longversion{The following characterization allows us to speak again about a box representation of~$\cG_{L_d^{\text{idim}}}$. This is in fact the same representation as investigated by Felsner, Habib and Möhring in \cite{FelHabMoh94}. In that paper, they showed that if one starts with a partial ordering $\prec$ on $V$, then its interval dimension equals the dimension of the partial ordering $\prec'$ defined between the boxes $B_v$ associated to the elements~$v$ of~$V$ via the intervals  $I_{v,i}$ by setting $B_u\prec' B_v$ \iffl $r_{u,i}\leq \ell_{v,i}$ for all $i\in [d]$. The ordering $\prec'$ can be also viewed as  a natural multi-dimensional counterpart of an interval ordering. However, this interrelation does not make the following theorem superfluous, because the gap between the interval dimension of a partial ordering $P$ and the dimension of~$P$ can be arbitrarily large, see~\cite{BogRabTro76}.}

\begin{thmrep}\applabel{thm:idimGraphs}
For each $d\in\mathbb{N}_{\geq 1}$, $\cG_{L_d^{\text{idim}}}$ is the class of  comparability graphs of posets of bounded interval dimension~$d$, where $L_{d}^{\text{idim}} = \{w \in \{0,1\}^{2d\text{-uni}}\mid  \exists x\in \langle0011\rangle\,\forall k\in [d]: \deletion{k}{k}{2d}(w) = x  \}$ is $2d$-uniform. \shortversion{Hence, $\cG_{L_d^{\text{idim}}}$  is factorial.}
\end{thmrep}
 
\begin{proof} 
Clearly, $L_{d}^{\text{idim}}$ is $2d$-uniform by construction.
\begin{enumerate}
\item Let $G = (V,E)$ be a comparability graph of posets of bounded interval dimension~$d$. By the multi-interval representation, there is a set of boxes $\{B_v\}_{v \in V}$ with $B_v = [\ell_1(v),r_1(v)] \times \dots \times [\ell_d(v),r_d(v)]$ for real numbers $\ell_j(v), r_j(v)$ for each $j \in [d]$ and $v \in V$. These boxes correspond to the vertices of $G$ in the sense that for all $u,v \in V$, $\{u,v\} \in E$ \iffl $r_j(u) < \ell_j(v)$ for all $j \in [d]$ or $r_j(u) > \ell_j(v)$ for all $j \in [d]$.
W.l.o.g., assume that all the $\ell_j(v)$ and $r_j(v)$ are distinct. Let $x_1 < \dots <x_{2d |V|}$ be the numbers $\{ \ell_j(v), r_j(v) \mid j \in [d], v \in V\}$ in linear order. We now define $w\in V^{2d|V|}$ by setting, for all $i \in [2d|v|]$, $w[i] \coloneqq v$ if there is a $j \in [d]$ such that $x_i \in \{ \ell_j(v), r_j(v)\}$. %Let $w \coloneqq \prod_{i \in [2d|v|]} w_i$. 
Note that for each $j \in [d]$ and $v \in V$, $\ell_j(v)=\position_v(2j-1,w)$ and $r_j(v)=\position_v(2j,w)$.
%correspond to the $(2j-1)$-th and the $2j$-th occurrence of $v$ in~$w$.
Therefore, for all $u,v\in V$, 
\begin{align*}
h_{u,v}(w) \in L_{d}^{\text{idim}} 
& \Leftrightarrow \left( \forall j\in [d]: \deletion{j}{j}{2d}(w) = 0011 \right) \vee \left( \forall j\in [d]: \deletion{j}{j}{2d}(w) = 1100 \right) \\
& \Leftrightarrow \left(\forall j\in [d]: r_j(u) < \ell_j(v)\right)  \vee \left(\forall j\in [d]: r_j(u) > \ell_j(v) \right)\\
& \Leftrightarrow \{u,v\} \in E.
\end{align*}
Hence, $G(L_{d}^{\text{idim}},w) = G$.
\item Let $G =(V,E)\in \mathcal{G}_{L_{d}^{\text{idim}}}$, i.e., there exists a word $w \in V^*$ describing~$G$, i.e., $G(L_{d}^{\text{idim}},w) = G$. 
By $2d$-uniformity, all $v \in \overline{V_{2d}(w)}$ are isolated vertices. 
For $v \in V_{2d}(w)$ and $j \in [d]$, define $I_{j}(v) \coloneqq [\ell_j(v),r_j(v)]$, where $\ell_j(v)$ is the $(2j-1)$-th and $r_j(v)$ the $2j$-th occurrence of $v$ in~$w$ and $B_v = [\ell_1(v),r_1(v)] \times \dots \times [\ell_d(v),r_d(v)]$. For each $v \in \overline{V_{2d}(w)}$, we choose a box $B_v$ that contains all other boxes. Clearly, $f$ is a $d$-dimensional box representation of~$G$. Proof details are similar to those leading to \myref{thm:lIntervalGraphs}.\qed 
\end{enumerate}
\end{proof}

\longversion{\begin{corollary}\label{cor:compBoundInt}
The class of  comparability graphs of posets of bounded interval dimension~$d$  is factorial for each fixed $d\in\mathbb{N}_{\geq 1}$.
\end{corollary}}

This speed question has been known for co-interval graphs, i.e., for $d=1$, but has been unknown for larger bounded interval dimensions~$d>1$. 
\longversion{This is explicitly stated for $d=2$ in \url{https://graphclasses.org/classes/speed.html}.}
%
%\todohf{The following might be an alternative way to trapezoid graphs.}
\longversion{According to}\shortversion{By}~\cite{Flo95}, cocomparability graphs of bounded interval dimension~$d$ are exactly the $d$-trapezoid graphs. We \longversion{can offer}\shortversion{have} an alternative proof for it\longversion{ now}: combine \myref{thm:idimGraphs} and \myref{thm:d-trapezoid-representation} with Proposition 4.3 and Theorem 4.14 of~\cite{FenFFMS2024}, as $L_{d}^{\text{idim}}=\{0,1\}^{2d\text{-uni}}\setminus L^{\text{trap}}_d$.

\section{Conclusions\longversion{ and Further Remarks}}

\longversion{In this paper, w}\shortversion{W}e mainly looked at the formalism of generalized word representability of graphs as  introduced in~\cite{FenFFMS2024}  from the perspective of settling several questions concerning the speed of graph classes qualified as \emph{unknown} or that are not even mentioned in graphclasses.org. 
In total, our approach resolves the factorial-versus-superfactorial question for about thirty graph classes.\begin{toappendix}
We only give few examples in the following\longversion{ (without presenting precise definitions of these classes)}: \emph{balanced 2-interval graphs} (they have an intersection model whose objects consist of two equal length intervals on a real line), or co-bipartite graphs of boxicity~2, or bipartite co-trapezoid graphs; these classes were known to have at least factorial speed by including graph classes with known factorial speed. For \emph{2-subdivision graphs}, graphclasses.org does not provide an entry concerning their speed. Clearly, the complement class of co-2-subdivision graphs is of the same speed. By \cite[Theorem~3]{FraGonOch2015},  co-2-subdivision graphs are 3-track graphs and hence they have at most factorial speed by \autoref{cor:l-track}. For \emph{$k$-leaf power graphs} (important in the context of phylogenetic trees, see~\cite{NisRagThi2002}), their factorial speed is only noted up to $k=4$ in graphclasses.org. However, by a result from~\cite{ChaFraMat2011}, $k$-leaf power graphs have boxicity at most $k-1$, so that our results prove that for each fixed~$k$, the class of $k$-leaf power graphs has (at most) factorial speed. Similarly, for graphs that allow for an embedding in the plane with at most $k$ crossings, it was shown in \cite{AdiChaMat2014} that then their boxicity is upper-bounded by a function in~$k$, so that we can deduce at most factorial speed for all graphs that allow a plane drawing  with at most $k$ crossings\longversion{, which somewhat confirms \myref{rem:planar}}.\longversion{\footnote{Notice that the argument in \myref{rem:planar} is based on relations of planar graphs to the poset dimension, while \cite{AdiChaMat2014} discusses relations to boxicity. In fact, there are several inter-relations between boxicity and poset dimension, see \cite{AdiBhoCha2011,FraGon2017}, but none of these seem to translate concrete factorial speed bounds between both notions.}}
\end{toappendix}
In passing, we also obtained quite a number of representations of graph classes\longversion{, in particular via $k$-uniform languages}. For instance, \longversion{a $d$-uniform language can characterize the $d$-dimensional comparability graphs, or }a $2b$-uniform language can characterize the graphs of boxicity~$b$. It is unclear if this is optimal. To formulate one concrete question: Is it possible to prove that no $(2b-1)$-uniform language can characterize the graphs of boxicity~$b$? 

\longversion{More generally speaking, the parameter~$k$ for languages~$L$ that are $k$-uniform might be an interesting parameter for the complexity of $\cG_L$ itself.\footnote{This corresponds to the parameter \emph{representation number} known from classical word-representable graphs, but there, mostly combinatorial aspects concerning single graphs have been considered, see \cite{GleKitPya2018,HalKitPya2011,Kit2013,Kit2017,KitPya2008}. Of course, one could also consider this parameter in the sense of a complexity measure for graphs, for each (infinite) 0-1-symmetric language~$L$ separately. This would be yet another direction of research.} For instance, for all 2-uniform 0-1-symmetric binary languages, we can observe that the recognition problem (given~$G$, is $G\in \cG_L$?) is solvable in polynomial time. Now, what is the smallest~$k$ such that the recognition problem for a graph class described by some $k$-uniform language is \NP-hard? For $\ell>1$, the question if a given graph is an $\ell$-interval graph is \NP-hard, see~\cite{WesShm84}. Hence, there is a 4-uniform language~$L$, namely,  $L=L_2^{\text{int}}$, such that the recognition problem for $\cG_L$ is \NP-hard.
According to \cite{GyaWes95,Jia2013a}, also the recognition problem for $\cG_{L_2^{\text{trk}}}$ is \NP-hard.
As another example, we could take $L=L_2^{\text{box}}$, based on~\cite{Kra94}. There are even 3-uniform languages with an \NP-hard recognition problem. According to \cite{Yan82}, $L_3^{\text{cmp}}$ is an example for this situation. We leave it as an open problem to understand the structure of $3$- or $4$-uniform languages~$L$ such that  the recognition problem for $\cG_L$ is \NP-hard. For recent complexity results on $k$-thin graphs, we refer to~\cite{Shi2025}.}

\longversion{Returning to the question of resolving the (unknown) speed question for graph classes, we should mention that there are still quite a number of remaining \emph{unknown} cases in graphclasses.org, and there are even cases when this nice database does not show any entry at all.
To mention some (related) examples, consider (co-)threshold tolerance (no entry), tolerance and leaf power graphs, or also strongly chordal or strongly orderable graphs. For strongly chordal graphs, we know that, similarly to chordal bipartite graphs, they contain graphs of arbitrarily large boxicity, see \cite{ChaFraMat2011,ChaFraMat2011a}, but while chordal bipartite graphs have superfactorial speed, this is open for strongly chordal graphs. However, this need not mean that these graph classes are superfactorial, as also circular-arc graphs contain graphs of arbitrarily large boxicity (see \cite{AdiBabCha2014}), but we have shown that their speed is only factorial.
}

We formulated our results on the speed of labeled graphs. In the literature, the ``unlabeled speed'' of graph classes is also considered, referring to the number of non-isomorphic graphs of order~$n$; see, e.g., \cite{BalBSS2009}. As the number of labeled graphs trivially upper-bounds the number of unlabeled ones and as the unlabeled speed of interval, of circle and of permutation graphs is known to be factorial, we\longversion{ immediately} obtain factorial unlabeled speed for most graph classes\longversion{ that we have} considered here.

\bibliographystyle{splncs04}
\bibliography{ab,hen,wrg-addendum}

\inappendixtrue
\appendix

\end{document}

%% file: shortcuts.tex
\usepackage[utf8]{inputenc}

\usepackage[]{todonotes}
\usepackage{hyperref} 
\usepackage{mathtools} 
\usepackage{amsmath,amsfonts,amssymb} 

\usepackage[capitalize]{cleveref}
\usepackage{complexity}
\usepackage{xspace}

\usetikzlibrary{positioning, shapes,patterns,shadows.blur, arrows,decorations,arrows,automata,shadows,patterns,chains,graphs,calc,intersections,matrix,fit,shapes,chains,decorations.pathreplacing,chains,fit,shapes}

\DeclareMathOperator{\classical}{cl}
\DeclareSymbolFont{Shuffle}{U}{shuffle}{m}{n}
\DeclareFontFamily{U}{shuffle}{}
\DeclareFontShape{U}{shuffle}{m}{n}{%
  <-8>shuffle7%
  <8->shuffle10%
}{}
\DeclareMathSymbol\shuffle{\mathbin}{Shuffle}{"001}
\DeclareMathSymbol\cshuffle{\mathbin}{Shuffle}{"002}

\newcommand{\emptyword}{\lambda}

\newcommand{\nth}[1]{\ensuremath#1^{\text{\tiny th}}}
\newcommand{\nst}[1]{\ensuremath#1^{\text{\tiny st}}}

\def\ta{\mathtt{a}}
\def\tb{\mathtt{b}}
\def\tc{\mathtt{c}}
\def\td{\mathtt{d}}
\def\te{\mathtt{e}}

\DeclareMathOperator{\alphabet}{alph}

\newcommand{\N}{\mathbb{N}}

\newcommand{\cC}{\mathcal{C}}
\newcommand{\cD}{\mathcal{D}}

\newcommand{\cG}{\mathcal{G}}

\newcommand{\cO}{\mathcal{O}}

\newcommand{\cT}{\mathcal{T}}
\newcommand{\cU}{\mathcal{U}}

\newcommand{\iffl}{\shortversion{iff}\longversion{if and only if}\xspace}

\longversion{\usepackage[appendix=inline]{apxproof}}
\shortversion{\usepackage{apxproof}}%
\shortversion{}

\newcommand{\todohf}[1]{\todo[color=red!80,inline]{#1}}